\documentclass[a4paper,fleqn]{article}
\usepackage{a4wide}
\usepackage{amsmath,amssymb,amsthm}
\usepackage{thmtools}
\usepackage{mathtools} 
\usepackage{enumerate}
\usepackage{abstract}
\usepackage{authblk}

\usepackage{todonotes}
\usepackage{hyperref}
\usepackage[T1]{fontenc}
\usepackage{mathpazo}
\makeatletter\@ifpackageloaded{mathpazo}\@tempswatrue\@tempswafalse
\if@tempswa
  \DeclareFontFamily{OT1}{pzc}{}
  \DeclareFontShape{OT1}{pzc}{m}{it}{<-> s * [1.15] pzcmi7t}{}
  \DeclareMathAlphabet{\mathpzc}{OT1}{pzc}{m}{it}
\fi\makeatother
\usepackage{bm}
\usepackage[utf8]{inputenc}

\usepackage[sorting=nyt,giveninits=true,backend=biber]{biblatex}
\makeatletter\@ifpackageloaded{biblatex}{%
  \usepackage{csquotes} 
  \bibliography{../../references}
  \renewbibmacro{in:}{%
    \ifentrytype{incollection}{\printtext{\bibstring{in}\intitlepunct}}{}}
  \renewbibmacro{publisher+location+date}{%
    \iflistundef{publisher}
      {\setunit*{\addcomma\space}}
      {\setunit*{\addcomma\space}}%
    \printlist{publisher}%
    \setunit*{\addcomma\space}%
    \printlist{location}%
    \setunit*{\addcomma\space}%
    \usebibmacro{date}%
    \newunit}
  \DeclareFieldFormat[article]{pages}{#1\isdot}
  \DeclareFieldFormat[article,incollection,inproceedings,unpublished,eprint]{title}{#1\isdot}
  \DeclareFieldFormat[thesis]{title}{\mkbibemph{#1\isdot}}
  \DeclareFieldFormat[unpublished]{date}{(#1)\isdot}
  \DeclareFieldFormat[unpublished]{note}{#1\nopunct} 
  \DeclareFieldFormat[eprint]{date}{(#1)\isdot}
  \DeclareFieldFormat[eprint]{note}{#1\nopunct} 
  \DeclareFieldFormat[article]{journaltitle}{\mkbibemph{#1\isdot}}
  
  \AtEveryBibitem{%
    \ifentrytype{book}{}{
      \clearname{editor}
    }
  }
  \newbibmacro*{bbx:parunit}{%
    \ifbibliography
      {\setunit{\bibpagerefpunct}\newblock
       \usebibmacro{pageref}%
       \clearlist{pageref}%
       \setunit{\adddot\par\nobreak}}
      {}
  }
  \renewbibmacro*{doi+eprint+url}{%
    \usebibmacro{bbx:parunit}
    \iftoggle{bbx:doi}
      {\printfield{doi}}
      {}%
    \iftoggle{bbx:eprint}
      {\usebibmacro{eprint}}
      {}%
    \iftoggle{bbx:url}
      {\usebibmacro{url+urldate}}
      {}
  }
  \renewbibmacro*{eprint}{%
    \usebibmacro{bbx:parunit}
    \iffieldundef{eprinttype}
      {\printfield{eprint}}
      {\printfield[eprint:\strfield{eprinttype}]{eprint}}
  }
  \renewbibmacro*{url+urldate}{%
    \usebibmacro{bbx:parunit}
    \printfield{url}%
    \iffieldundef{urlyear}
      {}
      {\setunit*{\addspace}%
       \printtext[urldate]{\printurldate}}
  }
}{}\makeatother

\declaretheorem[numberwithin=section,refname={theorem,theorems},Refname={Theorem,Theorems}]{theorem}
\declaretheorem[sibling=theorem,style=definition]{definition}
\declaretheorem[sibling=theorem,name=Lemma]{lemma}
\declaretheorem[sibling=theorem,name=Proposition]{proposition}
\declaretheorem[sibling=theorem,name=Corollary]{corollary}
\declaretheorem[sibling=theorem,style=definition,name=Remark]{remark}
\declaretheorem[sibling=theorem,style=definition,name=Example]{example}

\declaretheorem[numbered=no,name=Open Problem]{problem}
\declaretheorem[numbered=no,name=Claim]{claim}

\makeatletter\@ifpackageloaded{hyperref}{%
  \usepackage{xcolor}
  \definecolor{dark-red}{rgb}{0.4,0.15,0.15}
  \definecolor{dark-blue}{rgb}{0.15,0.15,0.4}
  \definecolor{medium-blue}{rgb}{0,0,0.5}
  \hypersetup{
    colorlinks,
    linkcolor={dark-red},
    citecolor={dark-blue},
    urlcolor={medium-blue}%
  }

}{}\makeatother

\usepackage{sectsty}
\sectionfont{\large\bfseries}
\subsectionfont{\normalsize\bfseries}

\newcommand{\Z}{\mathbb{Z}}
\newcommand{\T}{\mathbb{T}}
\providecommand{\abs}[1]{\lvert#1\rvert}

\newcommand{\infw}[1]{%
  \ifcat\noexpand#1\relax\bm{#1}
  \else\mathbf{#1}\fi}          
\newcommand{\Lang}[2][]{\mathcal{L}_{#1}(#2)}
\newcommand{\mirror}[1]{\widetilde{#1}}
\newcommand{\LO}[1]{\overline #1 }
\newcommand{\g}{\gamma}
\newcommand{\oa}{\mathfrak{a}}
\newcommand{\ob}{\mathfrak{b}}
\newcommand{\oc}{\mathfrak{c}}
\newcommand{\lexleq}{\lhd}
\newcommand{\lexgeq}{\rhd}

\newcommand{\keywords}[1]{\par\noindent{\footnotesize{\em Keywords\/}: #1}}

\begin{document}
  \title{More on the dynamics of the symbolic square root map}
  \author[,1,2]{Jarkko Peltomäki\footnote{Corresponding author.\\E-mail addresses: \href{mailto:r@turambar.org}{r@turambar.org} (J. Peltomäki), \href{mailto:mawhit@utu.fi}{mawhit@utu.fi} (M. Whiteland).}}
  \affil[1]{Turku Centre for Computer Science TUCS, Turku, Finland}
  \author[2]{Markus Whiteland}
  \affil[2]{University of Turku, Department of Mathematics and Statistics, Turku, Finland}
  \date{}
  \maketitle
  \vspace{-1.5em}
  \noindent
  \hrulefill
  \begin{abstract}
    \vspace{-1em}
    \noindent
    In our earlier paper [A square root map on Sturmian words, Electron. J. Combin. 24.1 (2017)], we introduced a
    symbolic square root map. Every optimal squareful infinite word $s$ contains exactly six minimal squares and can be
    written as a product of these squares: $s = X_1^2 X_2^2 \cdots$. The square root $\sqrt{s}$ of $s$ is the infinite
    word $X_1 X_2 \cdots$ obtained by deleting half of each square. We proved that the square root map preserves the
    languages of Sturmian words (which are optimal squareful words). The dynamics of the square root map on a Sturmian
    subshift are well understood. In our earlier work, we introduced another type of subshift of optimal squareful
    words which together with the square root map form a dynamical system. In this paper, we study these dynamical
    systems in more detail and compare their properties to the Sturmian case. The main results are characterizations of
    periodic points and the limit set. The results show that while there is some similarity it is possible for the
    square root map to exhibit quite different behavior compared to the Sturmian case.

    \vspace{1em}
    \keywords{sturmian word, optimal squareful word, symbolic square root map}
    \vspace{-1em}
  \end{abstract}
  \hrulefill

  \section{Introduction}
  Kalle Saari showed in \cite{diss:kalle_saari,2010:everywhere_alpha-repetitive_sequences_and_sturmian_words} that
  every Sturmian word contains exactly six minimal squares (that is, squares having no proper square prefixes) and that
  each position of a Sturmian word begins with a minimal square. Thus a Sturmian word $\infw{s}$ can be expressed as a
  product of minimal squares: $\infw{s} = X_1^2 X_2^2 \cdots$. In our earlier work
  \cite{2017:a_square_root_map_on_sturmian_words}, see also \cite{diss:jarkko_peltomaki}, we defined the square root
  $\sqrt{\infw{s}}$ of the word $\infw{s}$ to be the infinite word $X_1 X_2 \cdots$ obtained by deleting half of each
  square $X_i^2$. We proved that the words $\infw{s}$ and $\sqrt{\infw{s}}$ have the same language, that is, the square
  root map preserves the languages of Sturmian words. More precisely, we showed that if $\infw{s}$ has slope $\alpha$
  and intercept $\rho$, then $\sqrt{\infw{s}}$ has intercept $\psi(\rho)$, where
  $\psi(\rho) = \frac12 (\rho + 1 - \alpha)$. The simple form of the function $\psi$ immediately describes the dynamics
  of the square root map in the subshift $\Omega_\alpha$ of Sturmian words of slope $\alpha$: all words in
  $\Omega_\alpha$ are attracted to the set $\{01\infw{c}_\alpha, 10\infw{c}_\alpha\}$ of words of intercept $1-\alpha$;
  here $\infw{c}_\alpha$ is the standard Sturmian word of slope $\alpha$.

  The square root map makes sense for any word expressible as a product of squares. Saari defines in
  \cite{2010:everywhere_alpha-repetitive_sequences_and_sturmian_words} an intriguing class of such infinite words which
  he calls optimal squareful words. Optimal squareful words are aperiodic infinite words containing the least number of
  minimal squares such that every position begins with a square. It turns out that such a word must be binary, and it
  must contain exactly six minimal squares; less than six minimal squares forces the word to be ultimately periodic.
  Moreover, the six minimal squares must be the minimal squares of some Sturmian language; the set of optimal squareful
  words is however larger than the set of Sturmian words. The six minimal squares of an optimal squareful word take the
  following form for some integers $\oa$ and $\ob$ such that $\oa \geq 1$ and $\ob \geq 0$:
    \begin{alignat*}{2}
      &0^2,                   && (10^\oa)^2, \\
      &(010^{\oa-1})^2, \quad && (10^{\oa+1}(10^\oa)^\ob)^2, \\
      &(010^\oa)^2,           && (10^{\oa+1}(10^\oa)^{\ob+1})^2.
    \end{alignat*}
  It is natural to ask if there are non-Sturmian optimal squareful words whose languages the square root map
  preserves. In \cite{2017:a_square_root_map_on_sturmian_words}, we proved by an explicit construction that such words
  indeed exist. The construction is as follows. The substitution
  \begin{equation*}
    \tau\colon
    \begin{array}{l}
      S \mapsto LSS \\
      L \mapsto SSS
    \end{array}
  \end{equation*}
  produces two infinite words $\infw{\Gamma}_1^* = SSSLSSLSS \cdots$ and $\infw{\Gamma}_2^* = LSSLSSLSS \cdots$
  having the same language $\mathcal{L}$. Let $\mirror{s}$ be a (long enough) reversed standard word in some Sturmian
  language and $L(\mirror{s}\,)$ be the word obtained from $\mirror{s}$ by exchanging its first two letters. By
  substituting the language $\mathcal{L}$ by the substitution $\sigma$ mapping the letters $S$ and $L$ respectively to
  $\mirror{s}$ and $L(\mirror{s}\,)$, we obtain a subshift $\Omega$ consisting of optimal squareful words. We proved
  that the words $\infw{\Gamma}_1$ and $\infw{\Gamma}_2$, the $\sigma$-images of $\infw{\Gamma}_1^*$ and
  $\infw{\Gamma}_2^*$, are fixed by the square root map and, more generally, either $\sqrt{\infw{w}} \in \Omega$ or
  $\sqrt{\infw{w}}$ is periodic for all $\infw{w} \in \Omega$.

  The aim of this paper is to study the dynamics of the square root map in the subshift $\Omega$ in the slightly
  generalized case where $\tau(S) = LS^{2\oc}$ and $\tau(L) = S^{2\oc+1}$ for some positive integer $\oc$ and to see in
  which ways the dynamics differ from the Sturmian case. Our main results are the characterization of periodic and
  asymptotically periodic points and the limit set. We show that asymptotically periodic points must be ultimately
  periodic points and that periodic points must be fixed points; there are only two fixed points: $\infw{\Gamma}_1$ and
  $\infw{\Gamma}_2$. We prove that any word in $\Omega$ that is not an infinite product of the words $\sigma(S)$ and
  $\sigma(L)$ must eventually be mapped to a periodic word, thus having a finite orbit, while products of $\sigma(S)$
  and $\sigma(L)$ are always mapped to aperiodic words. It follows that the limit set contains exactly the words that
  are products of $\sigma(S)$ and $\sigma(L)$. Additionally, we show that the limit set can be expressed as a disjoint
  union of infinitely many invariant subsets. Moreover, we study the injectivity of the square root map on $\Omega$:
  only certain left extensions of the words $\infw{\Gamma}_1$ and $\infw{\Gamma}_2$ may have more than one preimage.

  Let us make a brief comparison with the Sturmian case to see that the obtained results indicate that the square root
  map behaves somewhat differently on $\Omega$. The mapping $\psi$, defined above, is injective, so in the Sturmian
  case all words have at most one preimage. As $\psi$ maps points strictly towards the point $1-\alpha$ on the circle,
  all points are asymptotically periodic (see \autoref{def:asymptotically_periodic}) and all periodic points are fixed
  points. The fixed points are the two words $01\infw{c}_\alpha$ and $10\infw{c}_\alpha$ mentioned above, and the
  limit set consists only of these two fixed points.

  This paper is an extended version of the conference paper
  \cite{2017:more_on_the_dynamics_of_the_symbolic_square_root} presented at the conference WORDS 2017 in Montr\'eal,
  September 2017. The conference paper omitted proofs, which are now given in this full paper. Further analysis on
  limit sets, namely the study of invariant sets, is presented in \autoref{sec:limit_set}. \autoref{sec:solutions}
  contains completely new material. Additionally Sections \ref{sec:further_remarks} and \ref{sec:open_problems} contain
  novel discussion on the topic.

  The paper is organized as follows. The following section gives needed preliminary results on Sturmian words and
  standard words, and it describes the construction of the subshift $\Omega$ in full detail. In \autoref{sec:limit_set},
  we proceed to characterize the limit set and to study injectivity. \autoref{sec:periodic_points} contains results on
  periodic points. Next, in \autoref{sec:solutions}, we take a closer look at solutions to a specific word equation
  that are important in our constructions. We conclude the paper by additional remarks in \autoref{sec:further_remarks}
  and open problems in \autoref{sec:open_problems}.

  \section{Notation and Preliminary Results}
  For basic word-combinatorial concepts, we refer the reader to the book \cite{2002:algebraic_combinatorics_on_words}
  or to the corresponding section of our previous paper \cite{2017:a_square_root_map_on_sturmian_words}. In this paper,
  we consider finite or infinite binary words, which we take to be over the alphabets $\{0, 1\}$ or $\{S, L\}$. We
  denote the length of the word $w$ by $\abs{w}$. The empty word has length $0$ and is denoted by $\varepsilon$. We
  refer to the $k^\text{th}$ letter of $w$ by $w[k]$, and we index letters from $0$. If $w = u^2$, then we call $w$ a
  \emph{square} with \emph{square root} $u$. A square is \emph{minimal} if it does not have a square as a proper
  prefix. By $L(w)$ we mean the word obtained from $w$ by exchanging its first two letters (we will not apply $L$ to
  too short words). Let $C$ be the cyclic shift operator defined by the formula
  $C(a_0 a_1 \cdots a_{n-1}) = a_1 \cdots a_{n-1} a_0$ for letters $a_i$. The words $w$, $C(w)$, $C^2(w)$, $\ldots$,
  $C^{\abs{w}-1}(w)$ are the \emph{conjugates of $w$}. If $u$ is a conjugate of $w$, then we say that $u$ is conjugate
  to $w$. We write $u \lexleq v$ if the word $u$ is lexicographically less than $v$. For binary words over $\{0,1\}$,
  we set $0 \lexleq 1$. If $w = uv$, then by $w \cdot v^{-1}$ we mean the word $u$.

  An infinite word is \emph{ultimately periodic} if it is of the form $uvvv\cdots$; otherwise it is \emph{aperiodic}.
  We distinguish finite words from infinite words by writing the symbols referring to infinite words in boldface. A
  \emph{subshift} $\Omega$ is a set of infinite words whose language is included in some extendable and factor-closed
  language $\Lang{\Omega}$, which is called the language of the subshift. If $\Lang{\Omega}$ is the language of some
  infinite word $\infw{w}$, then we say that the corresponding subshift is generated by $\infw{w}$. Subshifts are
  clearly shift-invariant; the shift operator on infinite words is denoted by $T$. If every word in a subshift is
  aperiodic, then we call the subshift \emph{aperiodic}. A subshift is \emph{minimal} if it does not contain nonempty
  subshifts as proper subsets.

  \subsection{Sturmian Words and Standard Words}
  Several proofs in \cite{2017:a_square_root_map_on_sturmian_words} regarding Sturmian words and the square root map
  require knowledge on continued fractions. In this paper, only some familiarity with continued fractions is required.
  We only recall that every irrational real number $\alpha$ has a unique infinite continued fraction expansion:
  \begin{equation}\label{eq:cf}
    \alpha = [a_0; a_1, a_2, a_3, \ldots] = a_0 + \dfrac{1}{a_1 + \dfrac{1}{a_2 + \dfrac{1}{a_3 + \ldots}}}
  \end{equation}
  with $a_0 \in \Z$ and $a_k \in \Z_+$ for $k \geq 1$. The numbers $a_i$ are called the \emph{partial quotients} of
  $\alpha$. The rational numbers $[a_0; a_1, a_2, a_3, \ldots, a_k]$, denoted by $p_k / q_k$, are called
  \emph{convergents} of $\alpha$. The \emph{semiconvergents} (or intermediate fractions) $p_{k,\ell} / q_{k,\ell}$ of
  $\alpha$ are defined as the fractions
  \begin{equation*}
    \frac{ \ell p_{k-1} + p_{k-2} }{ \ell q_{k-1} + q_{k-2} }
  \end{equation*}
  for $1 \leq \ell < a_k$ and $k \geq 2$ (if they exist). An introduction to continued fractions in relation to
  Sturmian words can be found in \cite[Chapter~4]{diss:jarkko_peltomaki}.

  We view here Sturmian words as the infinite words obtained as codings of orbits of points in an irrational circle
  rotation with two intervals. For alternative definitions and further details, see
  \cite{2002:substitutions_in_dynamics_arithmetics_and_combinatorics,2002:algebraic_combinatorics_on_words}. We
  identify the unit interval $[0,1)$ with the unit circle $\T$. Let $\alpha$ in $(0,1)$ be irrational. The map
  $R\colon \T \to \T, \, \rho \mapsto \{\rho + \alpha\}$, where $\{\rho\}$ stands for the fractional part of the number
  $\rho$, defines a rotation on $\T$. Divide the circle $\T$ into two intervals $I_0$ and $I_1$ defined by the points
  $0$ and $1-\alpha$. Then define the coding function $\nu$ by setting $\nu(\rho) = 0$ if $\rho \in I_0$ and
  $\nu(\rho) = 1$ if $\rho \in I_1$. The coding of the orbit of a point $\rho$ is the infinite word
  $\infw{s}_{\rho,\alpha}$ obtained by setting its $n^\text{th}, n \geq 0,$ letter to equal $\nu(R^n(\rho))$. This word
  $\infw{s}_{\rho,\alpha}$ is defined to be the Sturmian word of slope $\alpha$ and intercept $\rho$. To make the
  definition proper, we need to define how $\nu$ behaves in the endpoints $0$ and $1-\alpha$. We have two options:
  either take $I_0 = [0,1-\alpha)$ and $I_1 = [1-\alpha,1)$ or $I_0 = (0,1-\alpha]$ and $I_1 = (1-\alpha,1]$. The
  difference is seen in the codings of the orbits of the points $\{-n\alpha\}$. This choice is largely irrelevant in
  this paper with the exception of the definition of the mapping $\psi$ in the next subsection. The only difference
  between Sturmian words of slope $[0;1,a_2,a_3,\ldots]$ and Sturmian words of slope $[0;a_2+1,a_3,\ldots]$ is that the
  roles of the letters $0$ and $1$ are reversed. We make the typical assumption that $a_1 \geq 2$ in \eqref{eq:cf}.

  Since the sequence $(\{n\alpha\})_{n\geq 0}$ is dense in $[0,1)$---as is well-known---Sturmian words of slope
  $\alpha$ have a common language denoted by $\Lang{\alpha}$. The Sturmian words of slope $\alpha$ form the Sturmian
  subshift $\Omega_\alpha$, which is minimal and aperiodic. Let $w$ denote a word $a_0 a_1 \cdots a_{n-1}$ of length
  $n$ in $\Lang{\alpha}$. Then there exists a unique subinterval $[w]$ of $\T$ such that $\infw{s}_{\rho,\alpha}$
  begins with $w$ if and only if $\rho \in [w]$. Clearly
  $[w] = I_{a_0} \cap R^{-1}(I_{a_1}) \cap \ldots \cap R^{-(n-1)}(I_{a_{n-1}})$. The points $0$, $\{-\alpha\}$,
  $\{-2\alpha\}$, $\ldots$, $\{-n\alpha\}$ partition the circle into $n+1$ subintervals which are in one-to-one
  correspondence with the words of $\Lang{\alpha}$ of length $n$. Arranging these $n+1$ points into increasing order
  gives an ordering of the level $n$ intervals: $I_0(n)$, $I_1(n)$, $\ldots$, $I_n(n)$. According to the following
  proposition, see \cite[Proposition~3.2]{2016:abelian_powers_and_repetitions_in_sturmian_words}, this ordering of the
  intervals arranges the associated factors into lexicographic order.

  \begin{proposition}\label{prp:sturmian_lexicographic_order}
    Let $n$, $i$, and $j$ be integers such that $0 \leq i,j \leq n$. Let $u, v \in \Lang{\alpha}$ be the factors of
    length $n$ such that $[u] = I_i(n)$ and $[v] = I_j(n)$. Then $u \lexleq v$ if and only if $i < j$.
  \end{proposition}

  Let $(d_k)$ be a sequence of positive integers. Corresponding to $(d_k)$, we define a sequence $(s_k)$ of
  \emph{standard words} by the recurrence
  \begin{align*}
    s_k = s_{k-1}^{d_k} s_{k-2}
  \end{align*}
  with initial values $s_{-1} = 1$, $s_0 = 0$. The sequence $(s_k)$ converges to an infinite word $\infw{c}_\alpha$,
  which is a Sturmian word of intercept $\alpha$ and slope $\alpha$, where $\alpha$ is an irrational with continued
  fraction expansion $[0;d_1+1,d_2,d_3,\ldots]$. Thus standard words related to the sequence $(d_k)$ are called
  standard words of slope $\alpha$. If $d_k = 1$ for all $k \geq 1$, then the associated standard words are called
  \emph{Fibonacci words}. The standard words are the basic building blocks of Sturmian words, and they have rich and
  surprising properties. For this paper, we only need to know that standard words are primitive and that the final two
  letters of a (long enough) standard word are different. Actually, in connection to the square root map, it is more
  natural to consider reversed standard words obtained by writing standard words from right to left. If $s$ is a
  standard word in $\Lang{\alpha}$, then also the reversed standard word $\mirror{s}$ is in $\Lang{\alpha}$ because
  $\Lang{\alpha}$ is closed under reversal. For more on standard words, see
  \cite[Chapter~2.2]{2002:algebraic_combinatorics_on_words}.

  \subsection{Optimal Squareful Words and the Square Root Map}
  An infinite word is \emph{squareful} if its every position begins with a square. An infinite word is \emph{optimal
  squareful} if it is aperiodic and squareful and it contains the least possible number of distinct minimal squares. In
  \cite{2010:everywhere_alpha-repetitive_sequences_and_sturmian_words}, Kalle Saari proves that optimal squareful words
  contain six distinct minimal squares; a squareful word containing at most five minimal squares is necessarily
  ultimately periodic. Moreover, Saari shows that optimal squareful words are binary and that the six minimal squares
  are of very restricted form. The square roots of the six minimal squares of an optimal squareful word are
  \begin{alignat}{2}\label{eq:min_squares}
    &S_1 = 0,                 && S_4 = 10^\oa, \nonumber \\
    &S_2 = 010^{\oa-1}, \quad && S_5 = 10^{\oa+1}(10^\oa)^\ob, \\
    &S_3 = 010^\oa,           && S_6 = 10^{\oa+1}(10^\oa)^{\ob+1}. \nonumber
  \end{alignat}
  for some integers $\oa$ and $\ob$ such that $\oa \geq 1$ and $\ob \geq 0$. We call an optimal squareful word
  containing the \emph{minimal square roots} of \eqref{eq:min_squares} an \emph{optimal squareful word with parameters
  $\oa$ and $\ob$}. Throughout this paper, we reserve this meaning for the fraktur letters $\oa$ and $\ob$.
  Furthermore, we agree that the symbols $S_i$ always refer to the minimal square roots of \eqref{eq:min_squares}.

  Let $\infw{s}$ be an optimal squareful word and write it as a product of minimal squares: 
  $\infw{s} = X_1^2 X_2^2 \cdots$ (such a product is unique). The \emph{square root} $\sqrt{\infw{s}}$ of $\infw{s}$ is
  the word $X_1 X_2 \cdots$ obtained by deleting half of each minimal square $X_i^2$. We reserve the notation
  $\sqrt[n]{\infw{s}}$ for the $n^\text{th}$ square root of $\infw{s}$. We chose this notation for its simplicity; the
  $n^\text{th}$ square root of a number $x$ would typically be denoted by $\sqrt[2^n]{x}$. We often consider square
  roots of finite words. We let $\Pi(\oa,\ob)$ to be the language of all nonempty words $w$ such that $w$ is a factor
  of some optimal squareful word with parameters $\oa$ and $\ob$ and $w$ is factorizable as a product of minimal
  squares \eqref{eq:min_squares}. Let $w \in \Pi(\oa,\ob)$, that is, $w = X_1^2 \cdots X_n^2$ for minimal square roots
  $X_i$. Then we can define the square root $\sqrt{w}$ of $w$ by setting $\sqrt{w} = X_1 \cdots X_n$. The square root
  map (on infinite words) is continuous with respect to the usual topology on infinite words (see
  \cite[Section~1.2.2.]{2002:algebraic_combinatorics_on_words}). The following lemma, used later, sharpens this
  observation.

  \begin{lemma}\label{lem:continuity}
    Let $\infw{u}$ and $\infw{v}$ be two optimal squareful words with the same parameters $\oa$ and $\ob$. If
    $\infw{u}$ and $\infw{v}$ have a common prefix of length $\ell$, then $\sqrt{\infw{u}}$ and $\sqrt{\infw{v}}$ have
    a common prefix of length $\lceil \ell/2 \rceil$.
  \end{lemma}
  \begin{proof}
    Say $\infw{u}$ and $\infw{v}$ have a nonempty common prefix $w$. We may suppose that $w \notin \Pi(\oa,\ob)$ as
    otherwise the claim is clear. Let $z$ be the longest prefix of $w$ that is in $\Pi(\oa,\ob) \cup \{\varepsilon\}$,
    and let $X^2$ and $Y^2$ respectively be the minimal square prefixes of the words $T^{\abs{z}}(\infw{u})$ and
    $T^{\abs{z}}(\infw{v})$. Hence $\sqrt{\infw{u}}$ begins with $\sqrt{z}X$ and $\sqrt{\infw{v}}$ begins with
    $\sqrt{z}Y$. Since $X$ and $Y$ begin with the same letter, it is easy to see that either $X$ is a prefix of $Y$ or
    $Y$ is a prefix of $X$. By symmetry, we suppose that $X$ is a prefix of $Y$. It follows that $\sqrt{\infw{u}}$ and
    $\sqrt{\infw{v}}$ have a common prefix of length $\abs{zX^2}/2$. By the maximality of $z$, we have
    $\abs{zX^2} > \abs{w}$ proving that $\sqrt{\infw{u}}$ and $\sqrt{\infw{v}}$ have a common prefix of length
    $\lceil \abs{w}/2 \rceil$.
  \end{proof}

  Sturmian words form a proper subset of optimal squareful words. If $\infw{s}$ is a Sturmian word of slope $\alpha$
  having continued fraction expansion as in \eqref{eq:cf}, then it is an optimal squareful word with parameters
  $\oa = a_1 - 1$ and $\ob = a_2 - 1$. The square root map is especially interesting for Sturmian words because it
  preserves their languages. Define a function $\psi\colon \T \to \T$ as follows. For $\rho \in (0,1)$, we set
  \begin{equation*}
    \psi(\rho) = \frac12 (\rho+1-\alpha),
  \end{equation*}
  and we set
  \begin{equation*}
    \psi(0) = \begin{cases}
                \frac12 (1-\alpha), &\text{if $0 \in I_0$,} \\
                1-\frac\alpha 2,    &\text{if $0 \notin I_0.$}
              \end{cases}
  \end{equation*}
  The mapping $\psi$ moves a point $\rho$ on $\T$ towards the point $1-\alpha$ by halving the distance between the
  points $\rho$ and $1-\alpha$. The distance to $1-\alpha$ is measured in the interval $I_0$ or $I_1$ depending on
  which of these intervals the point $\rho$ belongs to. In \cite{2017:a_square_root_map_on_sturmian_words}, we proved
  the following result relating the intercepts of a Sturmian word and its square root.

  \begin{theorem}\label{thm:square_root}
    Let $\infw{s}_{\rho,\alpha}$ be a Sturmian word of slope $\alpha$. Then
    $\sqrt{\infw{s}_{\rho,\alpha}} = \infw{s}_{\psi(\rho),\alpha}$.
  \end{theorem}

  \begin{remark}\label{rem:periodic_sturmian_system}
    Later in the proof of \autoref{thm:periodic_finite_time}, we need a version of \autoref{thm:square_root} for
    rational slopes. Indeed, \autoref{thm:square_root} is true also if $\alpha$ is rational provided that the continued
    fraction expansion of $\alpha$ has enough partial quotients. More precisely, the proof of \autoref{thm:square_root}
    in \cite{2017:a_square_root_map_on_sturmian_words} considers only certain properties of the denominators of the
    (semi)convergents $q_{2,1}$ and $q_{3,1}$ of $\alpha$. Thus if the continued fraction expansion of $\alpha$ has at
    least three partial quotients, then \autoref{thm:square_root} holds. Moreover, this condition on the continued
    fraction expansion guarantees that the codings of rational rotations of slope $\alpha$ contain the six minimal
    squares of \eqref{eq:min_squares}.
  \end{remark}

  Specific solutions to the word equation
  \begin{equation}\label{eq:square}
    X_1^2 X_2^2 \cdots X_n^2 = (X_1 X_2 \cdots X_n)^2
  \end{equation}
  in the Sturmian language $\Lang{\alpha}$ play an important role. We are interested only in the solutions of
  \eqref{eq:square} where all words $X_i$ are minimal square roots \eqref{eq:min_squares}. Thus we give the following
  definition.

  \begin{definition}
    A nonempty word $w$ is a \emph{solution to \eqref{eq:square}} if $w$ can be written as a product of minimal square
    roots $w = X_1 X_2 \cdots X_n$ which satisfy the word equation \eqref{eq:square}. The solution is \emph{primitive}
    if $w$ is primitive. The word $w$ is a \emph{solution to \eqref{eq:square} in a language $\mathcal{L}$} if $w$ is a
    solution to \eqref{eq:square} and $w^2 \in \mathcal{L}$.
  \end{definition}

  Consider for example the word $S_2 S_1 S_4$ for $\oa = 1$ and $\ob = 0$. We have
  \begin{equation*}
    (S_2 S_1 S_4)^2 = (01 \cdot 0 \cdot 10)^2 = 01010 \cdot 01010 = (01)^2 \cdot 0^2 \cdot (10)^2 = S_2^2 S_1^2 S_4^2,
  \end{equation*}
  so the word $S_2 S_1 S_4$ is a solution to \eqref{eq:square}.

  In \cite[Theorem~18]{2017:a_square_root_map_on_sturmian_words}, the following result was proved.

  \begin{theorem}\label{thm:standard_solution}
    If $\mirror{s}$ is a reversed standard word, then the words $\mirror{s}$ and $L(\mirror{s}\,)$ are primitive
    solutions to \eqref{eq:square}.
  \end{theorem}

  Solutions to \eqref{eq:square} are important as they can be used to build fixed points of the square root map. If 
  $(u_k)$ is a sequence of solutions to \eqref{eq:square} with the property that $u_k^2$ is a proper prefix of
  $u_{k+1}$ for $k \geq 1$, then the infinite word $\infw{w}$ obtained as the limit $\lim_{k\to\infty} u_k$ has
  arbitrarily long prefixes $X_1^2 \cdots X_n^2$ with the property that $X_1 \cdots X_n$ is a prefix of $\infw{w}$. In
  other words, the word $\infw{w}$ is a fixed point of the square root map. All known constructions of fixed points
  rely on this method. For example, the two Sturmian words $01\infw{c}_\alpha$ and $10\infw{c}_\alpha$ of slope
  $\alpha$ and intercept $1-\alpha$ both have arbitrarily long squares $u^2$ as prefixes, where $u = L(\mirror{s}\,)$
  for a reversed standard word $\mirror{s}$ \cite[Proposition~27]{2017:a_square_root_map_on_sturmian_words}. In the
  next subsection, we see that the dynamical system studied in this paper is also fundamentally linked to fixed points
  obtained from solutions of \eqref{eq:square}.

  The following lemma \cite[Lemma~21]{2017:a_square_root_map_on_sturmian_words} is of technical nature, but it conveys
  an important message: under the assumptions of the lemma, swapping two adjacent and distinct letters that do not
  occur as a prefix of a minimal square affects a product of minimal squares only locally and does not change its
  square root. This establishes the often-used fact that $\mirror{s}\mirror{s}$ and $\mirror{s} L(\mirror{s}\,)$ are
  both in $\Pi(\oa,\ob)$ and have the same square root for a reversed standard word $\mirror{s}$. For example, if
  $\mirror{s} = 1001001010010$, then
  \begin{align*}
    \mirror{s}\mirror{s}      &= 1001001010 \cdot 0101 \cdot 00 \cdot 1001010010 \quad \text{and} \\
    \mirror{s}L(\mirror{s}\,) &= 1001001010 \cdot 010010 \cdot 1001010010,
  \end{align*}
  so the change is indeed local and does not affect the square root. Notice that every long enough standard word has
  $S_6$ as a proper suffix.

  \begin{lemma}\label{lem:exchange_squares}
    Let $u$ and $v$ be words such that
    \begin{itemize}
      \item $u$ is a nonempty suffix of $S_6$,
      \item $\abs{v} \geq \abs{S_5 S_6}$,
      \item $v$ begins with $xy$ for distinct letters $x$ and $y$,
      \item $uv$ and $L(v)$ are factors of some optimal squareful words with the same parameters.
    \end{itemize}
    Suppose there exists a minimal square $X^2$ such that $\abs{X^2} > \abs{u}$ and $X^2$ is a prefix of $uv$ or
    $uL(v)$. Then there exist minimal squares $Y_1^2$, $\ldots$, $Y_n^2$ such that $X^2$ and $Y_1^2 \cdots Y_n^2$ are
    prefixes of $uv$ and $uL(v)$ of the same length and $X = Y_1 \cdots Y_n$.
  \end{lemma}

  \subsection{The Subshift \texorpdfstring{$\Omega$}{Omega}}\label{ssec:subshift_omega}
  In this subsection, we define the main object of study of this paper. The results presented were obtained in
  \cite{2017:a_square_root_map_on_sturmian_words} in the case $\oc = 1$, the generalization being straightforward.

  Let $\oc$ be a fixed positive integer. Repeated application of the substitution
  \begin{equation*}
    \tau\colon
    \begin{array}{l}
      S \mapsto LS^{2\oc} \\
      L \mapsto S^{2\oc+1}
    \end{array}
  \end{equation*}
  to the letter $S$ produces two infinite words
  \begin{align*}
    \infw{\Gamma}_1^* &= SS^{2\oc} (LS^{2\oc})^{2\oc} (S^{2\oc+1} (LS^{2\oc})^{2\oc} )^{2\oc} \cdots \text{ and} \\
    \infw{\Gamma}_2^* &= LS^{2\oc} (LS^{2\oc})^{2\oc} (S^{2\oc+1} (LS^{2\oc})^{2\oc} )^{2\oc} \cdots
  \end{align*}
  with the same language $\mathcal{L}$. We set $\Omega^*$ to be the minimal and aperiodic subshift with language
  $\mathcal{L}$.

  Fix integers $\oa$ and $\ob$ such that $\oa \geq 1$ and $\ob \geq 0$, and let $\alpha$ be an irrational with
  continued fraction expansion $[0;\oa+1,\ob+1,\ldots]$. Let $w$ to be a word such that
  $w \in \{\mirror{s}_k, L(\mirror{s}_k)\}$ where $\mirror{s}_k$ is a reversed standard word of slope $\alpha$
  such that $\abs{\mirror{s}_k} > \abs{S_6}$.\footnote{Without this condition the subshift $\Omega$, defined below, does not consist of optimal squareful words; see the remark after \cite[Lemma~40]{2017:a_square_root_map_on_sturmian_words}.}
  Let then $\sigma$ be the substitution mapping $S$ to $w$ and $L$ to
  $L(w)$. By substituting the letters $S$ and $L$ in words of $\Omega^*$, we obtain a new minimal and aperiodic
  subshift $\sigma(\Omega^*)$, which we denote by $\Omega_A$. We also set
  $\infw{\Gamma}_1 = \sigma(\infw{\Gamma}_1^*)$ and $\infw{\Gamma}_2 = \sigma(\infw{\Gamma}_2^*)$. The subshift
  $\Omega_A$ is generated by both of the words $\infw{\Gamma}_1$ and $\infw{\Gamma}_2$. The words $\infw{\Gamma}_1$ and
  $\infw{\Gamma}_2$ differ only by their first two letters. This difference is often irrelevant to us, so we let
  $\infw{\Gamma}$ to stand for either of these words. Further, we let the symbol $\g_k$ to stand for the word
  $\sigma(\tau^k(S))$ and $\LO{\g_k}$ to stand for $\sigma(\tau^k(L))$.
  
  It is easy to see that $\infw{\Gamma}_1 = \lim_{k\to\infty} \g_{2k}$ and $\infw{\Gamma}_2 = \lim_{k\to\infty}
  \LO{\g_{2k}}$. In what follows, we often consider infinite products of $\g_k$ and $\LO{\g_k}$, and we wish to argue
  independently of the index $k$. Hence we make a convention that $\g$ and $\LO{\g}$ respectively stand for $\g_k$ and
  $\LO{\g_k}$ for some $k \geq 0$. The words $\g$ and $\LO{\g}$ are primitive; see
  \cite[Lemma~39]{2017:a_square_root_map_on_sturmian_words}. For simplification, we abuse notation and write $S$ for
  $\gamma_0$ and $L$ for $\LO{\gamma_0}$. It will always be clear from context if letters $S$ and $L$ or words $S$ and
  $L$ are meant.

  It can be shown that the words of $\Omega_A$ are optimal squareful words with parameters $\oa$ and $\ob$; see
  \cite[Lemma~40]{2017:a_square_root_map_on_sturmian_words}. Therefore the square root map is defined for words in
  $\Omega_A$. Let us prove the following crucial properties of the square root map on $\Omega_A$.

  \begin{lemma}\label{lem:crucial_properties}
    The following properties hold:
    \begin{itemize}
      \item $\sqrt{\g \g} = \g$,
      \item $\sqrt{\g \LO{\g}} = \g$,
      \item $\sqrt{\LO{\g} \g} = \LO{\g}$, \ and
      \item $\sqrt{\LO{\g} \LO{\g}} = \LO{\g}$.
    \end{itemize}
  \end{lemma}
  \begin{proof}
    This proof is essentially the proof of \cite[Proposition~38]{2017:a_square_root_map_on_sturmian_words}. Say
    $\g = \g_k$. Suppose first that $k = 0$. Since $S$ was chosen to be in $\{\mirror{s}, L(\mirror{s}\,)\}$ for a
    reversed standard word $\mirror{s}$, both of the words $S$ and $L$ are primitive solutions to \eqref{eq:square} by
    \autoref{thm:standard_solution}. Therefore $S^2, L^2 \in \Pi(\oa,\ob)$ and $\sqrt{SS} = S$ and $\sqrt{LL} = L$. An
    application of \autoref{lem:exchange_squares} shows that also $SL, LS \in \Pi(\oa,\ob)$ and that $\sqrt{SL} = S$
    and $\sqrt{LS} = L$.\footnote{\autoref{lem:exchange_squares} is indeed applicable: if $S$ (or $L$) was in $\Pi(\oa,\ob)$, then it would not be primitive due to the fact that it is a solution to \eqref{eq:square}.}
    Thus the claim holds for $k = 0$. Suppose that the claim holds for some $k \geq 0$. Now
    \begin{equation*}
      \g_{k+1}^2 = \LO{\g}\g \cdot (\g^2)^{\oc-1} \cdot \g\LO{\g} \cdot (\g^2)^\oc
    \end{equation*}
    so, by the induction hypothesis, we obtain
    \begin{equation*}
      \sqrt{\g_{k+1}^2} = \LO{\g} \cdot \g^{\oc-1} \cdot \g \cdot \g^\oc = \g_{k+1}.
    \end{equation*}
    The other cases are verified similarly by grouping the words into suitable pairs.
  \end{proof}

  \autoref{lem:crucial_properties} shows that the words $\infw{\Gamma}_1$ and $\infw{\Gamma}_2$ are fixed points of the
  square root map. Namely, the word $\gamma_{k+2}$ has $\gamma_k^2$ as a prefix and $\LO{\gamma_{k+2}}$ has
  $\LO{\gamma_k}^2$ as a prefix. Thus by \autoref{lem:crucial_properties}, we have, e.g.,
  \begin{equation*}
    \sqrt{\infw{\Gamma}_1} = \sqrt{\lim_{k\to\infty} \gamma_{2k}^2} = \lim_{k\to\infty} \gamma_{2k} = \infw{\Gamma}_1.
  \end{equation*}

  The words in $\Omega_A$ can be (uniquely) written as a product of the words $S$ and $L$ up to a shift. We often
  consider infinite words that are arbitrary products of the words $S$ and $L$ (elements of $\{S, L\}^\omega$) and
  their shifts (elements of the shift orbit closure $\overline{\{S, L\}^\omega}$). Consider a word $\infw{w}$ in
  $\overline{\{S, L\}^\omega}$ and write $\infw{w} = T^\ell({\infw{w'}})$ for some $\infw{w'} \in \{S, L\}^\omega$ and
  $\ell$ such that $0 \leq \ell < \abs{S}$. There are four distinct possibilities (types):
  \begin{enumerate}[(A)]
    \item $\ell = 0$, \label{type:a}
    \item $\ell > 0$ and the prefix of $\infw{w}$ of length $\abs{S}-\ell$ is in $\Pi(\oa,\ob)$, \label{type:b}
    \item $\ell > 0$ and the prefix of $\infw{w}$ of length $2\abs{S}-\ell$ is in $\Pi(\oa,\ob)$, or \label{type:c}
    \item none of the above apply. \label{type:d}
  \end{enumerate}
  These possibilities are mutually exclusive: cases \eqref{type:b} and \eqref{type:c} cannot simultaneously apply
  because $S, L \notin \Pi(\oa,\ob)$. In our earlier paper, we proved the following theorem, see
  \cite[Theorem~44]{2017:a_square_root_map_on_sturmian_words}.\footnote{In the proof of \cite[Theorem~44]{2017:a_square_root_map_on_sturmian_words} only the case $\oc = 1$ was considered, but the proof generalizes to the case $\oc > 1$ in a straightforward manner.}

  \begin{theorem}\label{thm:nr_periodic}
    Let $\infw{w} \in \Omega_A$. If $\infw{w}$ is of type \eqref{type:a}, \eqref{type:b}, or \eqref{type:c}, then
    $\sqrt{\infw{w}} \in \Omega_A$. If $\infw{w}$ is of type \eqref{type:d}, then $\sqrt{\infw{w}}$ is periodic with
    minimal period conjugate to $S$.
  \end{theorem}

  The next result is a direct consequence of the proof of \cite[Theorem~44]{2017:a_square_root_map_on_sturmian_words}.

  \begin{theorem}\label{thm:nr_periodic_general}
    Let $\infw{w} \in \overline{\{S, L\}^\omega}$. If $\infw{w}$ is of type \eqref{type:d}, then $\sqrt{\infw{w}}$ is
    periodic with minimal period conjugate to $S$.
  \end{theorem}

  Thus to make $\Omega_A$ a proper dynamical system, we need to adjoin a periodic part to it. To this end, we let
  \begin{equation*}
    \Omega_P = \{T^\ell(S^\omega)\colon 0 \leq \ell < \abs{S}\},
  \end{equation*}
  and define $\Omega = \Omega_A \cup \Omega_P$. Related to $\Omega_P$, we observe the following. Suppose that $S$ or
  $L$ equals $\mirror{s}_k$, the $k^\text{th}$ reversed standard word of slope $\alpha$, where
  $\alpha = [0;a_1,a_2,\ldots]$. Let us truncate the continued fraction expansion of $\alpha$ and set
  $\overline{\alpha} = [0;a_1,a_2,\ldots,a_k]$, so that $\overline{\alpha} = p_k/q_k$ with $\abs{S} = q_k$. Since
  $\{-q_k \overline{\alpha}\} = 0$, we have that $\Omega_P$ equals the codings of rational rotations of slope
  $\overline{\alpha}$. (See the proof of Theorem 4.3 and Example 4.4 in
  \cite{2015:characterization_of_repetitions_in_sturmian_words_a_new} for the exact details.)

  Clearly $\Omega$ is compact and $\sqrt{\Omega_A} \subseteq \Omega$ by
  \autoref{thm:nr_periodic}. On the other hand, for any $w\in \Omega_P$, we have $\sqrt{w}\in \Omega_P$ as noted
  in \autoref{rem:periodic_sturmian_system} (recall that $S$ was chosen so that it satisfies $\abs{S}>\abs{S_6}$).
  Thus $\sqrt{\Omega_P} \subseteq \Omega_P$, and the pair $(\Omega, \sqrt{\cdot})$ is a valid dynamical system. Notice
  further that $L^\omega \in \Omega_P$; it is a special property of a reversed standard word $\mirror{s}$ that
  $\mirror{s}$ and $L(\mirror{s}\,)$ are conjugates, see \cite[Proposition~6]{2017:a_square_root_map_on_sturmian_words}.

  Let us recall next what is known about the structure of the words in $\Omega$. The word $\infw{\Gamma}$ is by
  definition an infinite product of the words $\g_k$ and $\LO{\g_k}$ for all $k \geq 0$. Thus all words in $\Omega_A$
  are (uniquely) factorizable as products of $\g_k$ and $\LO{\g_k}$ up to a shift. Let us for convenience denote by
  $\Omega_\g$ the set $\Omega \cap \{\g, \LO{\g}\}^\omega$ consisting of words of $\Omega$ that are infinite products
  of $\g$ and $\LO{\g}$. Notice that $\Omega_S = \Omega_{\g_0}$ by our convention. The following lemma describes two
  important properties of factorizations of words of $\Omega_A$ as products of $\g$ and $\LO{\g}$. This result is an
  immediate property of the substitution $\tau$ that generates $\Omega^*$.

  \begin{lemma}\label{lem:factorization_properties}
    Consider a factorization of a word in $\Omega_A \cap \Omega_\g$ as a product of $\g$ and $\LO{\g}$. Such
    factorization has the following properties:
    \begin{itemize}
      \item Between two occurrences of $\LO{\g}$ there is always $\g^{2\oc}$ or $\g^{4\oc+1}$.
      \item Between two occurrences of $\LO{\g}\g^{4\oc+1}\LO{\g}$ there is always $\g^{2\oc}$ or
            $(\g^{2\oc}\LO{\g})^4 \cdot {\LO{\g}}^{-1}$.
    \end{itemize}
  \end{lemma}

  We also need to know how certain factors synchronize or align in a product of $\g$ and $\LO{\g}$. The proof is a
  straightforward application of the elementary fact that a primitive word cannot occur nontrivially in its square.

  \begin{lemma}[Synchronizability Properties]\label{lem:synchronization}
    Let $\infw{w} \in \Omega_\g$. If $z$ is a word in $\{\g\g, \g\LO{\g}, \LO{\g}\g\}$ occurring at position $\ell$ of
    $\infw{w}$, then the prefix of $\infw{w}$ of length $\ell$ is a product of $\g$ and $\LO{\g}$.\footnote{In general, e.g, the word $\g^2$ can be a factor of $\LO{\g}^3$.}
  \end{lemma}

  The preceding lemma shows that if $\infw{w}$ is a word in $\Omega_A$, then for each $k$ there exists a unique $\ell$
  such that $0 \leq \ell < \abs{\g_k}$ and $T^\ell(\infw{w}) \in \Omega_{\g_k}$. We then say that the
  $\g_k$-factorization of $\infw{w}$ starts at the position $\ell$ of $\infw{w}$.

  Let us conclude this subsection by making a remark regarding the subshift $\Omega^*$. It is possible to define a
  counterpart for the square root map of $\Omega$. Write a word $\infw{w}$ of $\Omega^*$ as a product of pairs of the
  letters $S$ and $L$: $\infw{w} = X_1 X_1' \cdot X_2 X_2' \cdots$, where $X_i X_i' \in \{SS, SL, LS, LL\}$. We define
  the square root $\sqrt{\infw{w}}$ of $\infw{w}$ to be the word $X_1 X_2 \cdots$. Based on the above, it is not
  difficult to see that $\sigma(\sqrt{\infw{w}}) = \sqrt{\sigma(\infw{w})}$ for $\infw{w} \in \Omega^*$. In other
  words, the square root map for words in $\Omega_S \cap \Omega_A$ has the same dynamics as the square root map in
  $\Omega^*$.

  \section{The Limit Set, Invariant Subsets, and Injectivity}\label{sec:limit_set}
  In this section, we consider what happens for words of $\Omega$ when the square root map is iterated. We extend
  \autoref{thm:nr_periodic_general} and show that also the words of type \eqref{type:b} and type \eqref{type:c} are
  eventually mapped to a periodic word. In fact, we prove a stronger result: the number of steps required is bounded by
  a constant depending only on the word $S$. These results enable us to characterize the limit set of
  $(\Omega, \sqrt{\cdot})$ as the set $\Omega_S$. In other words, asymptotically the square root map on $\Omega$ has
  the same dynamics as the counterpart mapping on $\Omega^* \cup \{S^\omega, L^\omega\}$. We further study invariant
  subsets and show that there are infinitely many of them. We also show that the square root map is mostly injective on
  $\Omega_A$, only certain left extensions of $\infw{\Gamma}$ may have two preimages.
  
  Let us first look at an example.

  \begin{example}
    Let $\oa = 1$, $\ob = 0$, and $S = 01010010$. Set $\infw{w} = T^4(S^2 \infw{u})$ for some
    $S^2 \infw{u} \in \Omega_S \cap \Omega_A$. The word $\infw{w}$ is of type \eqref{type:c} as the word $T^4(S^2)$,
    which equals $00 \cdot 1001010010$, is in $\Pi(\oa, \ob)$. Now $\sqrt{\infw{w}} = 010010 \cdot \sqrt{\infw{u}}$ and
    $\sqrt{\infw{w}} \in \Omega_A$ by \autoref{thm:nr_periodic}. So $\sqrt{\infw{w}}$ is of type \eqref{type:b}, and
    $\sqrt[2]{\infw{w}} = 010 \cdot \sqrt[2]{\infw{u}}$. Still we have $\sqrt[2]{\infw{w}} \in \Omega_A$. It is clear
    now that $\sqrt[2]{\infw{w}}$ is not of type \eqref{type:a} or \eqref{type:b}. The word $\sqrt[2]{\infw{u}}$ begins
    with $S$ or $L$, and neither $010 \cdot S$ nor $010 \cdot L$ is in $\Pi(\oa,\ob)$, so $\sqrt[2]{\infw{w}}$ is not
    of type \eqref{type:c} either. Thus it is of type \eqref{type:d}, so $\sqrt[3]{\infw{w}}$ is periodic. The minimal
    period of $\sqrt[3]{\infw{w}}$ is readily checked to be $01010010$, that is, $\sqrt[3]{\infw{w}} = S^\omega$. With
    some effort it can be verified that in this particular case $\sqrt[3]{\infw{v}}$ is periodic for all
    $\infw{v} \in \Omega \setminus \Omega_S$.
  \end{example}

  Notice that the parameter $\oc$ is irrelevant to all of the arguments in the above example. Notice also that the word
  $\infw{u}$ did not play any special role here, and it could have been any product of the words $S$ and $L$. Indeed,
  we formulate the next result for arbitrary products of the words $S$ and $L$.

  \begin{theorem}\label{thm:periodic_finite_time}
    There exists an integer $n$, depending only on the word $S$, such that
    $\sqrt[n]{\infw{w}} \in \{S^\omega, L^\omega\}$ for all
    $\infw{w} \in \overline{\{S, L\}^\omega} \setminus \{S, L\}^\omega$.
  \end{theorem}

  Before proceeding to prove the theorem, let us remark that the number $n$ in the statement indeed varies when $S$
  varies. Let $S$ be a reversed Fibonacci word. \autoref{tbl:fib_periodic} shows how $\abs{S}$ relates to $n$. It seems
  that here $n \to \infty$ as $\abs{S} \to \infty$ even though the growth is slow. We have not attempted seriously to
  relate $\oa$, $\ob$, $S$, and $n$, but we conjecture that $n$ is close to the bound obtained in
  \autoref{rem:bound_iterations_to_periodic}; see the final section on open problems.

  \begin{table}
  \centering  
  \begin{tabular}{c *{16}{@{\hspace{3.1mm}}c}}
    $\abs{S}$ \hspace{2mm} & 8 & 13 & 21 & 34 & 55 & 89 & 144 & 233 & 377 & 610 & 987 & 1597 & 2584 & 4181 & 6765 \\
    \hline \\
    $n$ \hspace{2mm}       & 3 & 4  & 4  & 5  & 6  & 6  & 7   & 8   & 8   & 9   & 10  & 10   & 11   & 12   & 13 \\
    \hline
  \end{tabular}
  \caption{How $\abs{S}$ and $n$ of \autoref{thm:periodic_finite_time} relate when $S$ is a reversed Fibonacci
word.}\label{tbl:fib_periodic}
  \end{table}

  For the proof, we need three lemmas. The first lemma is the important \nameref{lem:embedding}.

  \begin{lemma}[Embedding Lemma]\label{lem:embedding}
    Let $\infw{w} \in \overline{\{S, L\}^\omega}$ and $u_1$ and $u_2$ to respectively be the prefixes of $\infw{w}$ and
    $\sqrt{\infw{w}}$ of length $\abs{S}$.
    \begin{enumerate}[(i)]
      \item If $\infw{w}$ begins with $0$ and $u_1 \neq u_2$, then $u_1 \lexleq u_2$.
      \item If $\infw{w}$ begins with $1$ and $u_1 \neq u_2$, then $u_1 \lexgeq u_2$.
    \end{enumerate}
  \end{lemma}
  \begin{proof}
    Suppose that $u_1 \neq u_2$, and let $v$ be the prefix of $\infw{w}$ of length $2\abs{S}$. The main idea of this
    proof is the following idea of embedding: we slightly modify the prefix $v$ so that it belongs to a Sturmian
    language and the square root of a Sturmian word beginning with this modified prefix has $u_2$ as a prefix. Then
    known properties of the lexicographic ordering of factors of Sturmian words together with the dynamics of the
    function $\psi$ prove the claim. More precisely, we want to find a word $v'$ with the following properties: $v'$
    belongs to a Sturmian language $\Lang{\alpha}$, $v'$ has $u_1$ as a prefix, $\abs{v} = \abs{v'}$, and the square
    root of any Sturmian word of slope $\alpha$ beginning with $v'$ has $u_2$ as a prefix. Once we establish the
    existence of a word $v'$ with such properties, the claim is then proved as follows. Let $\rho \in [v']$, and
    consider the word $\infw{s}_{\rho,\alpha}$ of intercept $\rho$ and slope $\alpha$. By \autoref{thm:square_root},
    the intercept of the word $\sqrt{\infw{s}}$ is $\psi(\rho)$. Since $\sqrt{\infw{s}}$ has $u_2$ as a prefix, we
    see that $\psi(\rho) \in [u_2]$. Because $\psi$ moves points towards the point $1-\alpha$ and $u_1 \neq u_2$,
    we see that the interval $[u_2]$ is strictly closer to $1-\alpha$ than the interval $[u_1]$ is. It thus follows
    from \autoref{prp:sturmian_lexicographic_order} that $u_1 \lexleq u_2$ if $\infw{w}$ begins with $0$ and
    $u_1 \lexgeq u_2$ if $\infw{w}$ begins with $1$.

    We shall fix $v'$ later on; for now, we construct two Sturmian languages in at least one of which the word $v'$
    occurs. Suppose that $S$ or $L$ equals $\mirror{s}_k$, the $k^\text{th}$ reversed standard word of slope $\alpha$,
    where $\alpha = [0;a_1,a_2,\ldots]$. Without loss of generality we may assume $a_{k+1} > 1$. Recall that the
    assumption $\abs{S} > \abs{S_6}$ implies that $k \geq 3$. We modify $\alpha$ to obtain two distinct slopes
    $[0;b_1,b_2,\ldots]$ and $[0;c_1,c_2,\ldots]$, respectively denoted by $\alpha_1$ and $\alpha_2$. We set
    $b_i = a_i$ and $c_i = a_i$ for $1 \leq i < k - 1$. Further we set $b_{k-1} = a_{k-1}$, $b_k = a_k$, and
    $b_{k+1} \geq 3$; the remaining partial quotients may be chosen arbitrarily. For $\alpha_2$, we set
    $c_{k-1} = a_{k-1} + 1$ if $a_k = 1$; otherwise we let $c_{k-1} = a_{k-1}$, $c_k = a_k - 1$, $c_{k+1} = 2$, and
    $c_{k+2} = 2$. Again, the remaining partial quotients are irrelevant. In the case of slope $\alpha_1$, by
    recalling that $\mirror{s}_i \mirror{s}_{i+1} = L(\mirror{s}_{i+1} \mirror{s}_i)$ for all $i \geq 0$, we see that
    \begin{equation*}
      \mirror{s}_{k+1} = \mirror{s}_{k-1} \mirror{s}_k^{\, b_{k+1}} = L(\mirror{s}_k)^{b_{k+1}} \mirror{s}_{k-1},
    \end{equation*}
    so $S^3, L^3 \in \Lang{\alpha_1}$. It is straightforward to show that $\mirror{s}_{k+1}^{\,2} \in \Lang{\alpha_1}$.
    Now
    \begin{equation*}
      \mirror{s}_{k+1}^{\,2} = (\mirror{s}_{k-1} \mirror{s}_k^{\,b_{k+1}})^2 = \mirror{s}_{k-1} \cdot \mirror{s}_k^{\,b_{k+1}} L(\mirror{s}_k) \cdot \mirror{s}_{k-1} \mirror{s}_k^{\,b_{k+1}-1},
    \end{equation*}
    so also $SSL \in \Lang{\alpha_1}$. Next we want to show that $LLS \in \Lang{\alpha_2}$. Suppose first that
    $a_k = 1$. Then $S = \mirror{s}_{k-2} \mirror{s}_{k-1}$. The $(k-1)^\text{th}$ standard word of slope $\alpha_2$
    now equals $L$:
    \begin{equation*}
      \mirror{s}_{k-3} \mirror{s}_{k-2}^{\,c_{k-1}} = L(\mirror{s}_{k-2}) \mirror{s}_{k-3} \mirror{s}_{k-2}^{\,a_{k-1}} = L(\mirror{s}_{k-2}) \mirror{s}_{k-1} = L(S).
    \end{equation*}
    Thus the above argument showing that $SSL \in \Lang{\alpha_1}$ now shows that $LLS \in \Lang{\alpha_2}$. Consider
    then the case $a_k > 1$. Let $e$ be the $(k+2)^\text{th}$ standard word of slope $\alpha_2$. By expanding it as
    a product of the $(k-1)^\text{th}$ and $(k-2)^\text{th}$ standard words, we see that
    \begin{equation*}
      e = \mirror{s}_{k-2} \mirror{s}_{k-1}^{\,c_k} ( \mirror{s}_{k-1} ( \mirror{s}_{k-2} \mirror{s}_{k-1}^{\,c_k}
      )^2 )^2.
    \end{equation*}
    As $e^2 \in \Lang{\alpha_2}$, it follows that $\mirror{s}_{k-1} e \in \Lang{\alpha_2}$. Now
    \begin{align*}
      \mirror{s}_{k-1} e &= \mirror{s}_{k-1} \mirror{s}_{k-2} \mirror{s}_{k-1}^{\,c_k} \cdot \mirror{s}_{k-1} ( \mirror{s}_{k-2} \mirror{s}_{k-1}^{\,c_k} )^2 \cdot \mirror{s}_{k-1} ( \mirror{s}_{k-2} \mirror{s}_{k-1}^{\,c_k} )^2 \\
                         &= \mirror{s}_{k-1} \mirror{s}_{k-2} \mirror{s}_{k-1}^{\,c_k + 1} \mirror{s}_{k-2} \mirror{s}_{k-1}^{\,c_k} \mirror{s}_{k-2} \mirror{s}_{k-1}^{\,c_k + 1} \cdot ( \mirror{s}_{k-2} \mirror{s}_{k-1}^{\,c_k} )^2 \\
                         &= L(\mirror{s}_{k-2}) \mirror{s}_{k-1}^{\,c_k + 1} \cdot L(\mirror{s}_{k-2}) \mirror{s}_{k-1}^{\,c_k + 1} \cdot \mirror{s}_{k-2} \mirror{s}_{k-1}^{\,c_k + 1} \cdot ( \mirror{s}_{k-2} \mirror{s}_{k-1}^{\,c_k} )^2 \\
                         &= L(S) L(S) S \cdot (\mirror{s}_{k-2} \mirror{s}_{k-1}^{\,c_k} )^2,
    \end{align*}
    so $LLS \in \Lang{\alpha_2}$.
    
    We now define $v'$ needed for the conclusion of the proof. Observe that the assumption $u_1 \neq u_2$ implies that
    $\infw{w}$ is not in $\{S, L\}^\omega$. Write $\infw{w} = T^\ell(\infw{w'})$ for some
    $\infw{w'} \in \{S, L\}^\omega$ and integer $\ell$ such that $0 < \ell < \abs{S}$. Moreover, for $t \geq 1$, we set
    $\diamond_t$ to be the prefix of $T^{(t-1)\abs{S}}(\infw{w'})$ of length $\abs{S}$; notice that
    $\diamond_t \in \{S, L\}$ for all $t \geq 1$. For now we make the assumption that $\ell > 1$; the case $\ell = 1$
    is handled at the end of the proof. This additional assumption gives us the freedom to substitute $\diamond_1$ by
    either of the words $S$ and $L$ without affecting $u_1$ or the prefix of $\sqrt{\infw{w}}$ of length $\abs{S}$. So
    we substitute $\diamond_1$ by $\diamond_2$, and select $v'$ to be the prefix of
    $T^\ell(\diamond_2 \! \diamond_2 \! \diamond_3)$ of length $2\abs{S}$. Observe that
    $\diamond_2 \! \diamond_2 \! \diamond_3 \in \{SSS, SSL, LLS, LLL\}$, so either $v' \in \Lang{\alpha_1}$ or
    $v' \in \Lang{\alpha_2}$. Clearly $v'$ has $u_1$ as a prefix and $\abs{v'} = \abs{v}$. \autoref{lem:continuity}
    implies that the square root of a Sturmian word having $v'$ as a prefix has $u_2$ as a prefix. Thus the word $v'$
    has the desired properties.

    Consider finally the case $\ell = 1$. From the definition of standard words, it is straightforward to show that the
    only possible minimal square prefixes of $S$ and $L$ are $S_2^2$ and $S_5^2$. Consequently, the word $u_1$ has
    $10^\oa 10^{\oa-1}$ or $0^{\oa+1}(10^\oa)^\ob 10^{\oa+1}(10^\oa)^\ob$ as a prefix. In the latter case, the word
    $u_2$ clearly has $0^{\lceil (\oa + 1)/2 \rceil}1$ as a prefix, so $u_1 \lexleq u_2$. Consider then the former
    case, and let $t$ be the largest integer such that $(10^\oa)^t$ is a prefix of $u_1$. With some effort, it can be
    shown that $u_2$ then has $(10^\oa)^{\lceil t/2 \rceil} 10^{\oa+1}$ as a prefix, that is, $u_1 \lexgeq u_2$.
  \end{proof}

  \begin{lemma}\label{lem:c_not_1}
    Let $w$ be any of the words $SS$, $SL$, $LS$, or $LL$. If $\ell$ is an odd integer such that
    $0 < \ell < \abs{S}$, then $T^\ell(w) \notin \Pi(\oa,\ob)$.
  \end{lemma}
  \begin{proof}
    Let $\ell$ be an odd integer such that $0 < \ell < \abs{S}$. Since $\abs{T^\ell(w)} = \abs{S^2} - \ell$, we see
    that $\abs{T^\ell(w)}$ is odd. Thus it is impossible that $T^\ell(w) \in \Pi(\oa,\ob)$.
  \end{proof}

  Using the two preceding lemmas, we can prove the next crucial result.

  \begin{lemma}\label{lem:monotone_or_periodic}
    Let $\infw{w} \in \overline{\{S, L\}^\omega} \setminus \{S, L\}^\omega$ and $u_1$, $u_2$, and $u_3$ to respectively
    be the prefixes of $\infw{w}$, $\sqrt{\infw{w}}$, and $\sqrt[2]{\infw{w}}$ of length $\abs{S}$.
    \begin{enumerate}[(i)]
      \item If $\infw{w}$ begins with $0$, then one of the following holds: $u_2 \lexgeq u_1$,
            $u_3 \lexgeq u_1$, or $\sqrt[2]{\infw{w}}$ is periodic.
      \item If $\infw{w}$ begins with $1$, then one of the following holds: $u_2 \lexleq u_1$,
            $u_3 \lexleq u_1$, or $\sqrt[2]{\infw{w}}$ is periodic.
    \end{enumerate}
  \end{lemma}
  \begin{proof}
    Let $\infw{w} \in \overline{\{S, L\}^\omega} \setminus \{S, L\}^\omega$, and write $\infw{w} = T^\ell(\infw{w'})$
    for some $\infw{w'} \in \{S, L\}^\omega$ and integer $\ell$ such that $0 < \ell < \abs{S}$. Let $u_1$, $u_2$, and
    $u_3$ respectively be the prefixes of $\infw{w}$, $\sqrt{\infw{w}}$, and $\sqrt[2]{\infw{w}}$ of length $\abs{S}$.
    Moreover, for $t \geq 1$, we set $\diamond_t$ to be the prefix of $T^{(t-1)\abs{S}}(\infw{w'})$ of length
    $\abs{S}$; notice that $\diamond_t \in \{S, L\}$ for all $t \geq 1$. We suppose for simplicity that $\infw{w}$
    begins with the letter $0$. The proof in the case that the first letter of $\infw{w}$ is $1$ is the same proof with
    the lexicographic orderings reversed. Since words of type \eqref{type:d} map to periodic words by
    \autoref{thm:nr_periodic_general}, we only need to consider words of type \eqref{type:b} and \eqref{type:c}.

    Suppose first that $\infw{w}$ is of type \eqref{type:b}. Our aim is to show that $u_2 \lexgeq u_1$. By definition,
    we have $\infw{w} = X_1^2 \cdots X_n^2 \cdot T^{\abs{S}}(\infw{w'})$ for some minimal squares $X_1^2$, $\ldots$,
    $X_n^2$. Observe that $\sqrt{\infw{w}}$ has the word $X_1 \cdots X_n \diamond_2$ as a prefix. If $\ell = 1$, then
    the argument in the last paragraph of the proof of the \nameref{lem:embedding} implies that $u_2 \lexgeq u_1$. Say
    $\ell > 1$. Now $u_1$ is a conjugate of $\diamond_2$ occurring at position $\ell$ of $\diamond_2 \diamond_2$. If
    $u_1 = u_2$, then $u_1$ also occurs at the position $\ell + \abs{X_1 \cdots X_n}$ of $\diamond_2 \diamond_2$. This
    is not possible as $\diamond_2$ is primitive. Thus we conclude that $u_1 \neq u_2$. By the \nameref{lem:embedding},
    we see that $u_2 \lexgeq u_1$.

    Suppose then that $\infw{w}$ is of type \eqref{type:c}. If $\sqrt{\infw{w}}$ is of type \eqref{type:b} then,
    by applying the arguments of the preceding paragraph to $\sqrt{\infw{w}}$, we see that either $u_2 \lexgeq u_1$ or
    $u_3 \lexgeq u_1$. If $\sqrt{\infw{w}}$ is of type \eqref{type:d}, then $\sqrt[2]{\infw{w}}$ is periodic. Thus we
    may focus on the case that $\sqrt{\infw{w}}$ is also of type \eqref{type:c}. We suppose that $u_1 = u_2$; otherwise
    $u_2 \lexgeq u_1$ by the \nameref{lem:embedding}. Because $\sqrt{\infw{w}}$ is of type \eqref{type:c},
    \autoref{lem:c_not_1} implies that $\ell > 2$. Since $u_1 = u_2$, the word $u_1$ occurs at position $\ell$ of
    $\diamond_2 \diamond_2$ and at position $\ell/2$ of $\diamond_3 \diamond_3$. Since $\diamond_2$ is primitive, we
    see that necessarily $\diamond_2 \neq \diamond_3$. For now, we make the additional assumption that $\ell \neq 4$.
    Suppose next on the contrary that $u_3 = u_2 = u_1$. Since $\infw{w}$ and $\sqrt{\infw{w}}$ are of type
    \eqref{type:c} and $\ell \neq 4$, it follows that $u_1$ occurs at position $\ell/4$ of $\diamond_5 \diamond_5$. Now
    either $\diamond_5 = \diamond_2$ or $\diamond_5 = \diamond_3$, so either $u_1$ respectively occurs at positions
    $\ell$ and $\ell/4$ of $\diamond_2 \diamond_2$ or $u_1$ respectively occurs at positions $\ell/2$ and $\ell/4$ of
    $\diamond_3 \diamond_3$. This contradicts the primitivity of $\diamond_2$ and $\diamond_3$. We conclude that
    $u_3 \neq u_2$, so $u_3 \lexgeq u_1$ by the \nameref{lem:embedding}. What is left is to consider the case
    $\ell = 4$. Write $\diamond_2 = abcdv$ and $\diamond_3 = bacdv$ for a word $v$ and letters $a$, $b$, $c$, and $d$
    such that $a \neq b$. Since $u_1$ occurs at position $4$ of $\diamond_2 \diamond_2$, we see that $u_1$ has $abcd$
    as a suffix. Further, we see that $u_1$ has $cd$ as a prefix and $ba$ as a suffix because it occurs at position $2$
    of $\diamond_3 \diamond_3$. Therefore $cd = ba$. Assume now for a contradiction that $u_3 = u_2 = u_1$. It follows
    that the prefix $b$ of $u_1$ must be followed by $cd$. However, this is a contradiction as $cd$ is a prefix of
    $u_1$ and $c \neq d$. Once again, we conclude that $u_3 \neq u_2$, that is, $u_3 \lexgeq u_2$.
  \end{proof}

  \begin{proof}[Proof of \autoref{thm:periodic_finite_time}]
    Let $\infw{w} \in \overline{\{S, L\}^\omega} \setminus \{S, L\}^\omega$ and $u_1$, $u_2$, and $u_3$ to respectively
    be the prefixes of $\infw{w}$, $\sqrt{\infw{w}}$, and $\sqrt[2]{\infw{w}}$ of length $\abs{S}$. Suppose $\infw{w}$
    begins with the letter $0$. By \autoref{lem:monotone_or_periodic}, one of the following holds: $u_2 \lexgeq u_1$,
    $u_3 \lexgeq u_1$, or $\sqrt[2]{\infw{w}}$ periodic. Since the prefixes of length $\abs{S}$ of the words in the
    orbit of $\infw{w}$ can increase lexicographically only finitely many times, it follows that $\infw{w}$ eventually
    gets mapped to a periodic word. The same conclusion holds if $\infw{w}$ begins with the letter $1$. We have proved
    that there exists an integer $n$, depending only on $\abs{S}$, such that $\sqrt[n]{\infw{w}}$ is periodic.

    Since $\sqrt[n]{\infw{w}}$ is periodic, it is a rotation word of rational slope $\overline{\alpha}$
    by \autoref{thm:nr_periodic_general}. By \autoref{rem:periodic_sturmian_system}, the function
    $\psi\colon \rho \mapsto \frac12 (\rho + 1 - \overline{\alpha})$ relates the intercepts of $\sqrt[n]{\infw{w}}$ and
    $\sqrt[n+1]{\infw{w}}$ (recall that $\abs{S} > \abs{S_6}$). As $\psi^i(\rho)$ tends to $1-\overline{\alpha}$ as
    $i \to \infty$, the word $\sqrt[n]{\infw{w}}$ eventually gets mapped to $S^\omega$ or $L^\omega$ as $[S]$ and $[L]$
    are the two intervals with endpoint $1-\overline{\alpha}$; for in-depth details see
    \cite[Section~4]{2017:a_square_root_map_on_sturmian_words}. As it clearly takes a bounded number of steps,
    depending only on $S$, for a point to map to $[S]$ or to $[L]$, the proof is complete.
  \end{proof}

  \begin{remark}\label{rem:bound_iterations_to_periodic}
    Observe that in the above system of rotation words of rational slope $\overline{\alpha}$, $\psi$ maps any point to
    $[S]$ in at most
    \begin{equation*}
      \left\lceil \log_2\left( \frac{1-\overline{\alpha}}{\min\{\abs{[S]},\abs{[L]}\}} \right) \right\rceil
    \end{equation*}
    many steps. Here $\abs{[S]}$ and $\abs{[L]}$ are respectively the geometric lengths of the intervals $[S]$ and
    $[L]$ (of slope $\overline{\alpha}$).
  \end{remark}

  Next we turn our attention to injectivity. The results provided next give sufficient information to characterize the
  limit set. There is a slight imperfection in the following results. Namely, we are unable to characterize the
  preimage of the periodic part $\Omega_P$, and we believe no nice characterization exists. First of all, the words
  $S^\omega$ and $L^\omega$ must have several preimages, periodic and aperiodic, by \autoref{thm:periodic_finite_time}.
  Secondly, if $\infw{w}$ in $\Omega_A$ is of type \eqref{type:d}, then not only is $\sqrt{\infw{w}}$ periodic with
  minimal period conjugate to $S$ but the square root of  any word in $\Omega_A$ that shares a prefix of length
  $3\abs{S}$ with $\infw{w}$ is periodic with the same minimal period.\footnote{See the proof of \cite[Theorem~44]{2017:a_square_root_map_on_sturmian_words} for precise details.}
  Therefore here we only focus on characterizing preimages of words in the aperiodic part $\Omega_A$.

  We begin with a lemma.

  \begin{lemma}\label{lem:helpful}
    Suppose that $\infw{u}$ and $\infw{v}$ are words in $\Omega_\g$ such that $\sqrt{\infw{u}} = \sqrt{\infw{v}}$. If
    $\infw{u} = \g \g \cdots$ and $\infw{v} = \g \LO{\g} \cdots$, then $\infw{u} = \g \g \g^{2\oc} \LO{\g} \cdots$ and
    $\infw{v} = \g \LO{\g} \g^{2\oc} \LO{\g} \cdots$ and both $\infw{u}$ and $\infw{v}$ must be preceded by
    $\LO{\g} \g^{2\oc-1}$ in $\Omega$.
  \end{lemma}
  \begin{proof}
    By \autoref{lem:factorization_properties}, the word $\infw{v}$ begins with $\g \LO{\g} \g^{2\oc}$. Suppose that
    $\infw{u}$ begins with $\g^t \LO{\g}$ for $t \geq 2$. Assume for a contradiction that $t$ is odd. Since the prefix
    $\g^t \LO{\g}$ of $\infw{u}$ is followed by $\g^{2\oc}$ by \autoref{lem:factorization_properties}, we see that
    $\sqrt{\infw{u}} = \g^{(t-1)/2}\g\g^{\oc} \cdots$ so, as $\sqrt{\infw{u}} = \sqrt{\infw{v}}$, we conclude that
    $\infw{v}$ begins with $\g \LO{\g} \g^{2\oc+1}$. Hence \autoref{lem:factorization_properties} implies that
    $\infw{v}$ has the word $\g \LO{\g} \g^{4\oc+1} \LO{\g} \g^{2\oc} \LO{\g} \g^{2\oc}$ as a prefix. Therefore
    $\sqrt{\infw{v}} = \g^{3\oc+2}\LO{\g} \cdots$. Since $\sqrt{\infw{u}} = \sqrt{\infw{v}}$, we see that the prefix of
    $\infw{u}$ of length $2(3\oc+2)\abs{\g}$ must be followed by $\LO{\g}$. Now by
    \autoref{lem:factorization_properties}, the distance between two occurrences of $\LO{\g}$ in $\infw{u}$ is always a
    multiple of $(2\oc+1)\abs{\g}$. Hence $t + r(2\oc+1) = 2(3\oc+2)$ for some positive integer $r$. Since $t \geq 2$,
    we see that $r \leq 2$. If $r = 2$, then $t = 2\oc + 2$, which is impossible as $t$ is odd. Thus the only option is
    that $r = 1$, that is, $t = 4\oc - 3$. We have thus concluded that
    $\infw{u} = \g^{4\oc-3} \LO{\g} \g^{2\oc} \LO{\g} \g^{2\oc}$. It follows that
    $\sqrt{\infw{u}} = \g^{3\oc-1} \LO{\g} \cdots$, which contradicts the assumption
    $\sqrt{\infw{u}} = \sqrt{\infw{v}}$. We have thus proved that $t$ must be even.

    Now $\sqrt{\infw{u}} = \g^{t/2} \LO{\g} \g^{2\oc} \cdots$ and
    $\sqrt{\infw{v}} = \g \g^{\oc} \cdots$. Since $\sqrt{\infw{u}} = \sqrt{\infw{v}}$, it must be that the prefix of
    $\infw{v}$ of length $t\abs{\g}$ must be followed by $\LO{\g}$. Like previously, we see that $t = 1 + r(2\oc + 1)$
    for some positive integer $r$. Since $t \leq 4\oc + 1$, we see that the only option is that $r = 1$, that is,
    $\infw{u} = \g \g \g^{2\oc} \LO{\g} \cdots$. Suppose next for a contradiction that $\infw{v}$ begins with
    $\g \LO{\g} \g^{2\oc} \g$. This means that $\sqrt{\infw{v}} = \g^{\oc+1} \g \cdots$. Now
    $\infw{u} = \g^{2\oc+2} \LO{\g} \cdots$, so $\sqrt{\infw{u}} = \g^{\oc+1} \LO{\g} \cdots$. Thus
    $\sqrt{\infw{u}} \neq \sqrt{\infw{v}}$; a contradiction. Thus we have shown that
    $\infw{u} = \g \g \g^{2\oc} \LO{\g} \cdots$ and $\infw{v} = \g \LO{\g} \g^{2\oc} \LO{\g} \cdots$. What is left is
    to show that both $\infw{u}$ and $\infw{v}$ must be preceded by $\LO{\g} \g^{2\oc-1}$.

    Since $\infw{u}$ begins with $\g^{2\oc+2}$, it is clear by \autoref{lem:factorization_properties} that it must be
    preceded by $\LO{\g} \g^{2\oc-1}$. Assume for a contradiction that $\infw{v}$ is preceded by $\g^{2\oc}$. By
    \autoref{lem:factorization_properties}, either $\infw{v}$ has $\g \LO{\g} (\g^{2\oc} \LO{\g})^4 \g^{4\oc+1}$ as a
    prefix or it has $\g \LO{\g} \g^{2\oc} \LO{\g} \g^{4\oc+1}$ as a prefix. Consider the former case, where
    $\sqrt{\infw{v}}$ begins with $\g^{\oc+1} (\LO{\g} \g^{2\oc})^2 \g^{2\oc+1}$. Clearly the prefix $\g^{2\oc+2}
    \LO{\g}$ of $\infw{u}$ must be followed by $\g^{2\oc} \LO{\g} \g^{2\oc}$. The square root of this prefix equals
    $\g^{\oc+1} \LO{\g} \g^{2\oc}$, so as $\sqrt{\infw{u}} = \sqrt{\infw{v}}$, we conclude by
    \autoref{lem:factorization_properties} that $\infw{u}$ has the word $\g^{2\oc+2} (\LO{\g} \g^{2\oc})^4 \LO{\g} \g$
    as a prefix. However, now $\sqrt{\infw{u}} = \g^{\oc+1} (\LO{\g} \g^{2\oc})^2 \LO{\g} \cdots \neq \sqrt{\infw{v}}$,
    which is impossible. Therefore we are left with the case where $\infw{v}$ begins with
    $\g \LO{\g} \g^{2\oc} \LO{\g} \g^{4\oc+1}$. Set $x = \g^{2\oc+2}(\LO{\g} \g^{2\oc})^2$,
    $y = \g \LO{\g} \g^{2\oc} \LO{\g} \g^{4\oc+1}$, $z = (\LO{\g} \g^{2\oc})^2$, and $\LO{z} = \LO{\g} \g^{4\oc+1}$.
    Observe that the prefixes $x$ and $y$ of $\infw{u}$ and $\infw{v}$ have the same length and that this length
    divided by $\abs{\g}$ is even. Further, notice that $\sqrt{x} = \sqrt{y} = \g^{\oc+1} \LO{\g} \g^{2\oc}$ and
    $\sqrt{z\vphantom{\cramped{\LO{z}}}} = \sqrt{\LO{z}} = \LO{\g} \g^{2\oc}$. Once again by applying
    \autoref{lem:factorization_properties}, we see that $\infw{v}$ must have $yz$ as a prefix. Since
    $\sqrt{\infw{u}} = \sqrt{\infw{v}}$, it follows that $xz$ is a prefix of $\infw{u}$. Now the prefix $xz$ of
    $\infw{u}$ must be followed by $\LO{z}$ implying that $yz^2$ is a prefix of $\infw{v}$. Thus the known prefixes of
    $\infw{u}$ and $\infw{v}$ (of the same length) end with $z\LO{z}$ and $z^2$. These suffixes must respectively be
    followed by $z$ and $\LO{z}$ yielding known suffixes $\LO{z}z$ and $z\LO{z}$. Now $z\LO{z}$ must be followed by $z$,
    and as $\sqrt{\infw{u}} = \sqrt{\infw{v}}$, the known suffix $\LO{z}z$ must be followed by $z$. One more similar
    argument shows that the pattern repeats: the next known suffixes must be $z\LO{z}$ and $z^2$. This shows that
    $\infw{u} = x(z\LO{z}zz)^\omega$ and $\infw{v} = y(zz\LO{z}z)^\omega$. Since $\infw{u}$ and $\infw{v}$ are
    ultimately periodic and in $\Omega_\g$, it must be that $\{\infw{u}, \infw{v}\} = \{S^\omega, L^\omega\}$. This is
    clearly impossible.
  \end{proof}

  The next theorem says that the square root map is not injective on $\Omega_A$ but that it is almost injective: only
  words of restricted form may have more than one preimage and even then there are at most two preimages. In the
  Sturmian case, all words have at most one preimage.

  \begin{theorem}\label{thm:injectivity}
    If $\infw{w}$ is a word in $\Omega_A$ having two preimages $\infw{u}$ and $\infw{v}$ in $\Omega$ under the square
    root map, then $\infw{u} = zS\infw{\Gamma}_1$ and $\infw{v} = zS\infw{\Gamma}_2$ where $zS$ is a suffix of some
    $\g_k$ such that $z \in \Pi(\oa,\ob)$.
  \end{theorem}
  \begin{proof}
    First of all, notice that $\sqrt{S\infw{\Gamma}_1} = \sqrt{S\infw{\Gamma}_2}$ because $\sqrt{SL} = \sqrt{SS}$. Thus
    we only need to show that words having two preimages must be of the claimed form. Assume that $\infw{u}$ and
    $\infw{v}$ are distinct words in $\Omega$ having the same square root in $\Omega_A$. Suppose first that $\infw{u}$
    and $\infw{v}$ are products of the words $S$ and $L$. Let next $\g = \g_k$ for some $k \geq 0$, and assume that the
    words $\infw{u}$ and $\infw{v}$ are products of $\g$ and $\LO{\g}$ and that they have a minimal common prefix. In
    other words, we have $\infw{u} = \g \g \cdots$ and $\infw{v} = \g \LO{\g} \cdots$. \autoref{lem:helpful} implies
    that $\infw{u} = \g \g \g^{2\oc} \LO{\g} \cdots$ and $\infw{v} = \g \LO{\g} \g^{2\oc} \LO{\g} \cdots$. In other
    words, $\infw{u} = \g \LO{\g_{k+1}} \cdots$ and $\infw{v} = \g \g_{k+1} \cdots$. Moreover, both $\infw{u}$ and
    $\infw{v}$ must be preceded by $\g_{k+1} \g^{-1}$. Thus the words $\g_{k+1} \g^{-1} \infw{u}$ and
    $\g_{k+1} \g^{-1} \infw{v}$ have prefixes $\g_{k+1} \LO{\g_{k+1}}$ and $\g_{k+1} \g_{k+1}$ respectively. Since
    $\abs{\g_{k+1}}/\abs{\g}$ is odd, we see that $\sqrt{\g_{k+1}\g^{-1}\infw{u}} = \sqrt{\g_{k+1}\g^{-1}\infw{v}}$.
    Therefore we can repeat our argument so far with $\g_{k+1}$ in place of $\g$. The conclusion is that
    $\infw{u} = \g \infw{\Gamma}_1$ and $\infw{v} = \g \infw{\Gamma}_2$ (or $\infw{u} = \g \infw{\Gamma}_2$ and
    $\infw{v} = \g \infw{\Gamma}_1$). Moreover, we have shown that $\infw{\Gamma}_1$ and $\infw{\Gamma}_2$ uniquely
    extend to the left by $\g_i$ for all $i \geq 1$. Therefore if we allow $\infw{u}$ and $\infw{v}$ to have
    arbitrarily long common prefix, it must be that $\infw{u} = zS\infw{\Gamma}_1$ and $\infw{v} = zS\infw{\Gamma}_2$,
    where $zS$ is a suffix of some $\g_k$ such that $z \in \Pi(\oa,\ob)$.
    
    Suppose then that one of the words $\infw{u}$ and $\infw{v}$ is not in $\Omega_S$. If $\infw{u}$ is not a product of
    $S$ and $L$, then neither can its square root be, so actually neither $\infw{u}$ nor $\infw{v}$ is in $\Omega_S$.
    Because words of type \eqref{type:d} map to periodic words by \autoref{thm:nr_periodic}, it must be that $\infw{u}$
    and $\infw{v}$ are of type \eqref{type:b} or \eqref{type:c}. Thus $\infw{u} = x \infw{w'}$ and
    $\infw{v} = y \infw{w''}$ for words $\infw{w'}$ and $\infw{w''}$ in $\Omega_S$ and words $x$ and $y$ such that
    $x, y \in \Pi(\oa,\ob)$ and $\abs{x}, \abs{y} < 2\abs{S}$. Now $\sqrt{\infw{u}} = \sqrt{\infw{v}}$ so, since
    $\abs{\sqrt{x}}, \abs{\sqrt{y}} < \abs{S}$ and infinite products of $S$ and $L$ synchronize, we conclude that
    $\sqrt{x} = \sqrt{y}$ and $\smash[t]{\sqrt{\infw{w'}} = \sqrt{\infw{w''}}}$. Thus by the arguments of the preceding
    paragraph, we have, say, $\infw{u} = x z S \infw{\Gamma}_1$ and $\infw{v} = yzS \infw{\Gamma}_2$ for some
    $z \in \{SS, SL, LS\}^*$. Since $\infw{\Gamma}$ uniquely extends to the left by $\g_i$ for all $i \geq 1$, we see
    that $\infw{u}$ and $\infw{v}$ are of the claimed form.
  \end{proof}

  Let us state separately an observation made in the proof of \autoref{thm:injectivity} that is helpful when we next
  characterize the points that are in the limit set.

  \begin{corollary}\label{cor:left_extension}
    In $\Omega$, the word $\infw{\Gamma}$ is uniquely extended to the left by $\g_k$ for all $k \geq 0$.
  \end{corollary}

  The \emph{limit set} $\Lambda$ is the set of words that have arbitrarily long chains of preimages, that is,
  \begin{equation*}
    \Lambda = \bigcap_{n = 0}^\infty \sqrt[n]{\Omega}.
  \end{equation*}
  In the Sturmian case, the limit set contains only the two fixed points of the square root map. For the subshift
  $\Omega$, the limit set is much larger. In fact, the limit set contains all words that are products of the words $S$
  and $L$. Proving this result is our next aim.

  \begin{theorem}\label{thm:limit_set}
    We have $\Lambda = \Omega_S$.
  \end{theorem}

  We begin with a lemma after which we proceed to prove \autoref{thm:limit_set}.

  \begin{lemma}\label{lem:help_preimage}
    If $\infw{w} \in \Omega_S$ and $\infw{w}$ has $\infw{\Gamma}$ as a suffix, then $\infw{w} \in \Lambda$.
  \end{lemma}
  \begin{proof}
    Since both $\infw{w}$ and $\infw{\Gamma}$ are in $\Omega_S$ and the factorization of a word as a product of the
    words $S$ and $L$ is unique, we see that $\infw{w} = z\infw{\Gamma}$ for some $z \in \{S, L\}^*$. Let $z'$ be a
    suffix of some $\g_k$ of length $2\abs{z}$. By \autoref{cor:left_extension}, we have $z'\infw{\Gamma} \in \Omega$.
    Now $\sqrt{z'\infw{\Gamma}} = \sqrt{z'}\infw{\Gamma}$ and $\abs{\sqrt{z'}} = \abs{z}$, so $\sqrt{z'} = z$ by
    \autoref{cor:left_extension}. Thus $\infw{w}$ has a preimage of the same form.
  \end{proof}

  \begin{proof}[Proof of \autoref{thm:limit_set}]
    Suppose first that $\infw{w} \in \Omega \setminus \Omega_S$. Since the square root of a word in $\Omega_S$ is also
    in $\Omega_S$, we see that all preimages of $\infw{w}$ are in $\Omega \setminus \Omega_S$. It is thus an immediate
    consequence of \autoref{thm:periodic_finite_time} that every backward orbit of $\infw{w}$ is finite, that is,
    $\infw{w} \notin \Lambda$.

    Suppose then that $\infw{w} \in \Omega_S$. The only periodic words in $\Omega_S$ are the fixed points $S^\omega$
    and $L^\omega$ which clearly have a preimage. It is thus enough to assume that $\infw{w}$ is aperiodic and to find
    two sequences $(u_n)$ and $(v_n)$ with the following properties:
    \begin{itemize}
      \item $u_n$ is a prefix of $\infw{w}$ for all $n \geq 1$ and the sequence $(\abs{u_n})$ is strictly increasing,
      \item $v_n \in \Lang{\Omega}$ and $\sqrt{v_n} = u_n$ for all $n \geq 1$.
    \end{itemize}
    By compactness, a subsequence of $(v_n)$ converges to an infinite word $\infw{v}$ in $\Omega$ with the property
    that $\sqrt{\infw{v}} = \infw{w}$. Therefore $\infw{w}$ has a preimage in $\Omega$ so, as $\infw{w}$ was arbitrary,
    we conclude that $\infw{w} \in \Lambda$.

    We may assume that $\infw{w}$ does not have $\infw{\Gamma}$ as a suffix by \autoref{lem:help_preimage}. Since
    $\infw{w} \neq \infw{\Gamma}$, there exists maximal $k_1$ such that the $\g_{k_1}$-factorization of $\infw{w}$
    starts at the beginning of $\infw{w}$. Let $j_1$ be the starting position of the $\g_{k_1+1}$-factorization of
    $\infw{w}$. By the maximality of $k_1$, we have $j_1 \in \{1, 2, \ldots, 2\oc\}\abs{\g_{k_1}}$. Let $u_1$ be the
    prefix of $\infw{w}$ of length $j_1$ and $v_1$ be the suffix of $\g_{k_1+1}^2$ of length $2\abs{u_1}$. We have
    $v_1 \in \Pi(\oa,\ob)$ by \autoref{lem:crucial_properties}, so $\sqrt{v_1} = u_1$ as $u_1$ is a suffix of
    $\g_{k_1+1}$. Again since $\infw{\Gamma}$ is not a suffix of $\infw{w}$, we see that there exists maximal $k_2$
    such that the $\g_{k_2}$-factorization of $\infw{w}$ starts at position $j_1$. The $\g_{k_2+1}$-factorization of
    $\infw{w}$ begins at position $j_2$ where $j_2 = j_1 + t\abs{\g_{k_2}}$ with $1 \leq t \leq 2\oc$. Set again $u_2$
    to be the prefix of $\infw{w}$ of length $j_2$ and $v_2$ to be the suffix of $\g_{k_2+1}^2$ of length $2\abs{u_2}$.
    By the definition of $j_2$, the word $u_2$ is a suffix of $\g_{k_2+1}$. Observe that $\abs{u_2}$ is a multiple of
    $\abs{\g_{k_1}}$, so $\abs{v_2}$ is an even multiple of $\abs{\g_{k_1}}$. Since $\g_{k_2+1}^2$ is a product of the
    words $\g_{k_1}$ and $\LO{\g_{k_1}}$, it follows that $\sqrt{v_2} = u_2$. Repeating these arguments, we obtain the
    desired sequences $(u_n)$ and $(v_n)$.
  \end{proof}

  Finally, we consider invariant subsets and show that the limit set $\Lambda$ is not simple in the sense that it
  contains infinitely many invariant subsets. We first show how to decompose $\Omega_\g$ into two invariant sets.

  Let $A_k$ be the set of words in $\Omega_{\g_k}$ that have one of the following words as a prefix:
  $\g_k \g_k^{2\oc} \LO{\g_k}$, $\LO{\g_k} \g_k^{2\oc} \LO{\g_k}$, or $\LO{\g_k} \g_k^{2\oc} \g_k$. Let us find the
  prefixes of length $(2\oc + 2)\abs{\g_k}$ of the preimages of the words in $A_k$. Let $\infw{w}$ be a word in
  $\Omega_{\g}$ (now $\g = \g_k$) such that $\sqrt{\infw{w}} \in A_k$. Say $\sqrt{\infw{w}}$ has
  $\LO{\g} \g^{2\oc} \LO{\g}$ as a prefix so that the prefix of $\infw{w}$ of length $2(2\oc+2)\abs{\g}$ is of the form
  $\LO{\g} \! \diamond \! (\g \diamond)^{2\oc} \LO{\g} \diamond$. By \autoref{lem:factorization_properties}, this
  prefix must equal $\LO{\g} \g^{2\oc} \! \diamond \! \g^{2\oc} \LO{\g} \g$ with $\diamond \in \{\g, \LO{\g}\}$, and it
  follows that $\infw{w} \in A_k$. Suppose next that $\infw{w}$ has a prefix of the form
  $\LO{\g} \! \diamond \! (\g \diamond)^{2\oc} \g \diamond$. Like previously, this prefix must take the form
  $\LO{\g} \g^{2\oc} (\diamond \g)^{\oc+1} \diamond$. If $\infw{w}$ has $\LO{\g} \g^{2\oc} \g$ as a prefix, then it has
  $\LO{\g} \g^{4\oc+1} \LO{\g}$ as a prefix by \autoref{lem:factorization_properties}. This is clearly a contradiction,
  so $\LO{\g} \g^{2\oc} \LO{\g}$ is a prefix of $\infw{w}$, which in turn implies that $\infw{w} \in A_k$. Consider the
  last case where $\infw{w}$ has a prefix of the form $\g \! \diamond \! (\g \diamond)^{2\oc} \LO{\g} \diamond$. Again,
  it must be that the prefix takes the form $(\g \diamond)^{\oc+1} \g^{2\oc} \LO{\g} \g$. If the prefix of length
  $(2\oc + 1)\abs{\g}$ is followed by $\g$, then $\infw{w}$ must begin with $\LO{\g}$ by
  \autoref{lem:factorization_properties}. As this is impossible, we see that again $\infw{w} \in A_k$. We have thus
  proved that $\Omega_\g \setminus A_k$ is invariant under the square root map. It is straightforward to see that
  also $A_k$ is invariant.

  Let us show next that $A_k = \Omega_{\g_{k+1}}$ for all $k \geq 0$. It is clear that
  $A_k \subseteq \Omega_{\g_{k+1}}$. Let $\infw{w} \in \Omega_{\g_{k+1}}$, and consider its prefix of length
  $(2\oc+2)\abs{\g_k}$. If $\LO{\g_k}$ begins at position $t\abs{\g_k}$ of $\infw{w}$ with $0 < t \leq 2\oc$, then
  clearly $\infw{w} \notin \Omega_{\g_{k+1}}$. If $\LO{\g_k}$ is a prefix of $\infw{w}$, then obviously
  $\infw{w} \in A_k$. If $\LO{\g_k}$ occurs at position $(2\oc + 1)\abs{\g_k}$ of $\infw{w}$, then $\infw{w}$ has
  either $\g_k \g_k^{2\oc} \LO{\g_k}$ or $\LO{\g_k} \g_k^{2\oc} \LO{\g_k}$ as a prefix, and $\infw{w} \in A_k$. Thus we
  are left with the case that $\g_k^{2\oc+2}$ is a prefix of $\infw{w}$. Now $\infw{w}$ has prefix
  $\g_k^{2\oc+2+r}\LO{\g_k}$ for some $r \geq 0$. As $\infw{w} \in \Omega_{\g_{k+1}}$, we see that $2\oc+2+r$ is a
  multiple of $2\oc+1$. Further by \autoref{lem:factorization_properties}, it must be that $2\oc+2+r = 2(2\oc+1)$,
  which implies that $r = 2\oc$. Thus $\infw{w}$ has $\g^{4\oc+2}$ as a prefix, which contradicts
  \autoref{lem:factorization_properties}.

  Putting together the results of the preceding two paragraphs yields the following result.

  \begin{proposition}\label{prp:invariant_subsets}
    We have the disjoint union
    \begin{equation*}
      \Lambda = \{\infw{\Gamma}_1, \infw{\Gamma}_2\} \cup \bigcup_{k=0}^\infty \Omega_{\g_k} \setminus \Omega_{\g_{k+1}}
    \end{equation*}
    of subsets invariant under the square root map.
  \end{proposition}
  \begin{proof}
    Notice that $\Lambda = \Omega_{\g_0}$ by \autoref{thm:limit_set}. Using the above arguments, we can write
    $\Omega_{\g_k} = A_k \cup (\Omega_{\g_k} \setminus A_k) = \Omega_{\g_{k+1}} \cup (\Omega_{\g_k} \setminus \Omega_{\g_{k+1}})$
    for all $k \geq 0$. The sets in the union are disjoint and invariant. Clearly the words in
    $\Lambda \setminus \bigcup_{k=0}^\infty \Omega_{\g_k} \setminus \Omega_{\g_{k+1}}$ are exactly the words whose
    $\g_k$-factorization begins at the beginning for all $k \geq 0$. These words are by construction the two fixed
    points $\infw{\Gamma}_1$ and $\infw{\Gamma}_2$. They clearly form an invariant subset.
  \end{proof}

  \section{Periodic Points}\label{sec:periodic_points}
  In this section, we characterize the periodic points of the square root map in $\Omega$. The result is that the only
  periodic points are fixed points. We further characterize asymptotically periodic points and show that all
  asymptotically periodic points are ultimately periodic points.

  Recall that a word $\infw{w}$ is a \emph{periodic point} of the square root map with period $n$ if
  $\sqrt[n]{\infw{w}} = \infw{w}$.
  
  \begin{theorem}\label{thm:periodic_points}
    If $\infw{w}$ is a periodic point in $\Omega$, then
    $\infw{w} \in \{\infw{\Gamma}_1, \infw{\Gamma}_2, S^\omega, L^\omega\}$.
  \end{theorem}
  \begin{proof}
    By \autoref{thm:periodic_finite_time}, no word in $\Omega \setminus \Omega_S$ can be a periodic point. Thus we
    assume that $\infw{w}$ is a word in $\Omega_S$ such that $\sqrt[n]{\infw{w}} = \infw{w}$ for some integer
    $n \geq 1$. Suppose for a contradiction that
    $\infw{w} \notin \{\infw{\Gamma}_1, \infw{\Gamma}_2, S^\omega, L^\omega\}$. It follows that there exists maximal
    $k$ such that the $\g_k$-factorization of $\infw{w}$ starts at the beginning of $\infw{w}$. Since the square root
    map acts essentially the same way on products of $\g$ and $\LO{\g}$ and on products of $S$ and $L$ due to
    \autoref{lem:crucial_properties}, we may assume that $k = 0$. For $i \geq 1$, let $k_i$ be such that the starting
    position of the $\g_i$-factorization of $\infw{w}$ equals $k_i\abs{S}$. In particular, we have $k_i \neq 0$ for all
    $i \geq 1$.

    Write $\infw{w} = a_0 a_1 \ldots$ for $a_t \in \{S, L\}$, and let $i \geq 1$. Define the infinite word $\infw{u}_i$
    as the subword $a_{k_i} a_{2k_i} a_{4k_i} \cdots a_{2^t k_i} \cdots$. Since $\sqrt[n]{\infw{w}} = \infw{w}$, we see
    that the relation $a_t = a_{2^n t}$ holds for all $t \geq 0$. Hence the word $\infw{u}_i$ has the property that
    $\infw{u}_i = T^{n\abs{S}}(\infw{u}_i)$ for all $i \geq 1$, i.e., it is purely periodic. We shall show that this is
    impossible, and thus that $\infw{w}$ does not exist.

    Because the $\g_i$-factorization of $\infw{w}$ begins at position $k_i\abs{S}$, we actually know most of the
    contents of the word $\infw{u}_i$ without knowing anything particular about the $\g_i$-factorization of $\infw{w}$.
    The word $\infw{u}_i$ is obtained by concatenating the factors of length $\abs{S}$ of $\g_i$ or $\LO{\g_i}$
    occurring at positions given by the sequence $(d_t\abs{S})$ where $d_t$ is given by the sequence
    $((2^t - 1)k_i)_{t\geq0}$ modulo $(2\oc+1)^i$. Thus there is ambiguity only when
    $(2^t-1)k_i \equiv 0 \pmod{(2\oc+1)^i}$. Let $p_i$ be the minimal period of the sequence $(d_t)$. We have
    \begin{equation}\label{eq:ui}
      \infw{u}_i = \prod_{t=1}^\infty \diamond_t v_i
    \end{equation}
    where $\diamond_t \in \{S, L\}$ and $v_i$ is a word of length $p_i - 1$ over $\{S,L\}$. Notice that $p_i > 1$
    because $k_i \neq 0$. Let us next see what the word $v_i$ is like.

    Suppose first that $i = 1$. Since $\g_1 = LS^{2\oc}$, we see that $v_i = S^{p_1-1}$. Let then $i > 1$, so we have
    $\g_i = \LO{\g_{i-1}} \g_{i-1}^{2\oc}$. Observe that $k_i \in k_{i-1} + \{0,\ldots,2\oc\}(2\oc+1)^{i-1}$ so, by
    basic modular arithmetic, it is straightforward to see that $p_{i-1}$ divides $p_i$. We see that the word $v_i$ has
    the word $v_{i-1}$ as a prefix since the prefix of $v_i$ of length $(p_i - 1)\abs{S}$ is determined by the
    positions $(2-1)k_i\abs{S}$, $(4-1)k_i\abs{S}$, $\ldots$, $(2^{p_{i-1}}-1)k_i\abs{S}$ of $\g_i$, with the
    coefficients of $\abs{S}$ taken modulo $(2\oc+1)^i$, that is, by the same positions of $\g_{i-1}$. Suppose then
    that $p_i > p_{i-1}$. By the form of $\g_i$, the next factor of $v_i$ of length $(p_i/p_{i-1})\abs{S}$ is
    determined by the positions $(1-1)k_i\abs{S}$, $(2-1)k_i\abs{S}$, $\ldots$, $(2^{p_{i-1}}-1)k_i\abs{S}$ of
    $\g_{i-1}$. Repeating this reasoning, it follows that $v_i = v_{i-1}(b v_{i-1})^{p_i/p_{i-1}-1}$ where $b = S$ if
    $\g_{i-1}$ begins with $S$ and $L$ otherwise. The words $v_i$ can be generated as follows. Let $v'_0 = S$. If
    $v'_i$ is defined, then $v'_{i+1} = L(v'_i) {v'_i}^{p_i/p_{i-1}-1}$. Now $v_i$ is obtained from $v'_i$ by deleting
    its first $\abs{S}$ letters. It is straightforward to show that the words $v'_i$ are primitive.

    Let us show next that the words $\diamond_t$ of \eqref{eq:ui} take both values $S$ and $L$ infinitely often.
    For this, we need to prove the following claim.

    \begin{claim}
      The sequence $(p_i)$ is increasing.
    \end{claim}
    \begin{proof}
      By the Chinese Remainder Theorem, it is sufficient to show that the sequence $(p_i)$ is increasing in the case
      that $2\oc+1 = p^\ell$ for a prime $p$. Suppose that $\gcd(k_1, 2\oc+1) = p^a$. Since $k_1 < 2\oc+1$, we have
      $a < \ell$. Suppose for a contradiction that $\gcd(k_i, (2\oc+1)^i) > p^a$ for some $i > 1$. Then $p^{a+1}$
      divides $k_i$. Since $k_i = k_{i-1} + rp^{\ell(i-1)}$ for some $r \in \{0,\ldots,2\oc\}$, it follows that
      $p^{a+1}$ divides $k_{i-1}$. Consequently, we see that $p^{a+1}$ divides $k_1$; a contradiction. Therefore
      $\gcd(k_i, (2\oc+1)^i) = p^a$ for all $i \geq 1$. Pick $j$ so large that $(2\oc+1)^j / p^a > 2^{p_i} - 1$. Then
      it must be that $(2^{p_i}-1)k_j \not\equiv 0 \pmod{(2\oc+1)^j}$, so $p_j > p_i$.
    \end{proof}

    The claim implies that there exists $j$ such that $p_{j+1} > p_j > p_i$. Since $p_j > p_i$, there exists infinitely
    many $t$ such that $(2^{tp_i}-1)k_i \equiv 0 \pmod{(2\oc+1)^i}$ and
    $(2^{tp_i}-1)k_i \not\equiv 0 \pmod{(2\oc+1)^j}$. Thus for these $t$, the word $\diamond_t$ equals the first
    $\abs{S}$ letters of the word $\g_{j-1}$. Similarly there exists infinitely many $t$ such that
    $(2^{tp_i}-1)k_i \equiv 0 \pmod{(2\oc+1)^i}$ and $(2^{tp_i}-1)k_i \not\equiv 0 \pmod{(2\oc+1)^{j+1}}$. For these
    numbers $t$, the word $\diamond_t$ equals the first $\abs{S}$ letters of the word $\g_j$. Since $\g_{j-1}$ begins
    with $S$ and $\g_j$ begins with $L$ or vice versa, the words $\diamond_t$ indeed take both values $S$ and $L$
    infinitely often.

    Let us show next that if $\infw{u}_i$ is purely periodic with period $m$, then $p_i\abs{S}$ divides $m$. Assume on
    the contrary that $p_i\abs{S}$ does not divide $m$. This means by \eqref{eq:ui} that for all large enough $r$ the
    word $\diamond_r v_i$ is an interior factor of $\diamond_s v_i \diamond_{s+1} v_i$ for some $s$. Let $j$ be the
    largest integer such that $\abs{v'_j} < \abs{v'_i}$. Due to the primitivity of the word $v'_j$, we see that the
    position where $\diamond_r v_i$ occurs at is a multiple of $\abs{v'_j}$. We conclude that $\diamond_r$ is uniquely
    determined by the first $\abs{S}$ letters of $v'_j$. This is a contradiction because $\diamond_t$ takes both values
    $S$ and $L$ infinitely often. Therefore $p_i\abs{S}$ divides $m$.

    Recall that, for all $i\geq 1$, $\infw{u}_i$ is purely periodic with period $n$. By the arguments of the previous
    paragraph, the number $p_i$ divides $n$ for all $i \geq 1$. This is absurd as the sequence $(p_i)$ is increasing.
    This contradiction concludes the proof.
  \end{proof}

  The case with the Sturmian periodic points is similar: periodic points are fixed points and the fixed points are
  obtained as limits from solutions of \eqref{eq:square}.


  Next we consider the dynamical notion of an asymptotically periodic point and characterize asymptotically periodic
  points in $\Omega$.

  \begin{definition}\label{def:asymptotically_periodic}
    Let $(X, f)$ be a dynamical system. A point $x$ in $X$ is \emph{asymptotically periodic} if there exists a periodic
    point $y$ in $X$ such that
    \begin{equation*}
      \lim_{n \to \infty} d(f^n(x), f^n(y)) = 0.
    \end{equation*}
    If this is the case, then we say that the point $x$ is \emph{asymptotically periodic to $y$}.
  \end{definition}

  The following proposition essentially says that if a word in $\Omega$ is asymptotically periodic, then it is an
  ultimately periodic point. The situation is opposite to the Sturmian case where all words are asymptotically periodic
  and only periodic points are ultimately periodic points.

  \begin{proposition}
    If $\infw{w} \in \Omega_S$, then $\infw{w}$ is asymptotically periodic if and only if
    $\infw{w} \in \{\infw{\Gamma}_1, \infw{\Gamma}_2, S^\omega, L^\omega\}$, that is, if and only if $\infw{w}$ is a
    periodic point. If $\infw{w} \in \Omega \setminus \Omega_S$, then $\infw{w}$ is asymptotically periodic to
    $S^\omega$ or $L^\omega$.
  \end{proposition}
  \begin{proof}
    Let $\infw{w} \in \Omega \setminus \Omega_S$. By \autoref{thm:periodic_finite_time}, there exists an integer $n$
    such that $\sqrt[n]{\infw{w}} \in \{S^\omega, L^\omega\}$, so $\infw{w}$ is asymptotically periodic to $S^\omega$
    or $L^\omega$. Suppose then that $\infw{w}$ in $\Omega_S$ is aperiodic and asymptotically periodic. By
    \autoref{thm:periodic_points}, this means that the sequence $(\sqrt[n]{\infw{w}})_n$ converges to $\infw{\Gamma}$.
    Observe that if $\infw{w} \notin \Omega_\gamma$, then also $\sqrt{\infw{w}} \notin \Omega_\gamma$. From the fact
    that $\infw{\Gamma} \in \Omega_{\gamma_k}$ for all $k \geq 0$ we thus conclude that
    $\infw{w} \in \Omega_{\gamma_k}$ for all $k \geq 0$ which means that $\infw{w} = \infw{\Gamma}$.
  \end{proof}

  \section{Solutions to the Word Equation in \texorpdfstring{$\Lang{\Omega}$}{Lang(Omega)}}\label{sec:solutions}
  This section contains a characterization of long enough solutions to the word equation \eqref{eq:square} in
  $\Lang{\Omega}$. The construction of the fixed points $\infw{\Gamma}_1$ and $\infw{\Gamma}_2$ introduces the
  solutions $S$, $L$, $\g_1$, $\g_2$, $\ldots$ into $\Lang{\Omega}$. The main result of this section,
  \autoref{thm:solution_characterization}, tells that these are essentially all solutions to \eqref{eq:square} in
  $\Lang{\Omega}$, the construction does not introduce any additional, or accidental, solutions.

  Let us first characterize squares in $\Omega^*$.

  \begin{lemma}\label{lem:square_conjugate}
    Let $u$ be primitive. Then $u^2 \in \Lang{\Omega^*}$ if and only if $u$ is conjugate to $\tau^k(S)$ for some
    $k \geq 0$.
  \end{lemma}
  \begin{proof}
    Observe that $(\tau^k(S))^3 \in \Lang{\Omega^*}$ for all $k \geq 0$ because $(\tau^k(S))^{4\oc+1}$ occurs between
    two occurrences of $\tau^k(L)$; see \autoref{lem:factorization_properties}. Therefore if $u$ is conjugate to
    $\tau^k(S)$, then $u^2 \in \Lang{\Omega^*}$.

    Suppose that $u^2 \in \Lang{\Omega^*}$ with $u$ primitive. If $u = \tau(v)$ for some word $v$, then
    $v^2 \in \Lang{\Omega^*}$ and, by induction, $v$ is conjugate to $\tau^k(S)$ for some $k \geq 0$. This means that
    $u$ must be conjugate to $\tau^{k+1}(S)$. Assume that $u$ is not of the form $\tau(v)$. If $u = S$, then the claim
    holds. Otherwise $u$ must contain at least one occurrence of the letter $L$, and it is possible to factorize
    $u = x\tau(u')ay$ for some words $u'$, $x$, and $y$ and letter $a$ such that $ayx$ equals $\tau(S)$ or $\tau(L)$.
    Let $b$ be a letter such that $b \neq a$. By the simple form of the substitution $\tau$, we see that
    $(u'b)^2 \in \Lang{\Omega^*}$. By induction, $u'b$ is conjugate to $\tau^k(S)$ for some $k \geq 0$. Consequently,
    $bu'$ is conjugate to $\tau^k(S)$, which in turn implies that $\tau(bu')$, which equals $ayx\tau(u')$, is conjugate
    to $\tau^{k+1}(S)$. Now $u$ and $ayx\tau(u')$ are conjugate, proving the claim.
  \end{proof}

  The next theorem, which is quite general, could be of independent interest in characterizing solutions to
  \eqref{eq:square} more generally.

  \begin{theorem}\label{thm:conjugate_solutions}
    If $u$ is a word that is a product of the words $S$ and $L$ and a primitive solution to \eqref{eq:square}, then
    none of its proper conjugates are solutions to \eqref{eq:square}, except in the case that $u \in \{S, L\}$ when $S$
    and $L$ are the only conjugates of $u$ that are solutions to \eqref{eq:square}.
  \end{theorem}
  \begin{proof}
    Suppose that $u$ is a word that is a product of the words $S$ and $L$ and a primitive solution to
    \eqref{eq:square}. If $u \in \{S, L\}$, then the only conjugates of $u$ that are solutions to \eqref{eq:square} are
    $S$ and $L$ by \cite[Theorem~18]{2017:a_square_root_map_on_sturmian_words}. We may thus suppose that
    $u \notin \{S, L\}$. Let $v$ be a proper conjugate of $u$, and assume for a contradiction that $v$ is a solution to
    \eqref{eq:square}.

    Suppose first that $v$ is a product of the words $S$ and $L$. Consider $u$ and $v$ as words over the alphabet
    $\{S, L\}$. The length of $u$ (as a word over $\{S, L\}$) must be odd, as otherwise we would have
    $u = (\sqrt{u})^2$ contradicting the primitivity of $u$. Since both $u$ and $v$ are solutions to \eqref{eq:square},
    we have $u[i] = u[2i]$ and $v[i] = v[2i]$ for $i \in \{0, 1, \ldots, \abs{u} - 1\}$ (here the indices are naturally
    interpreted modulo $\abs{u}$). Since $v$ is a proper conjugate of $u$, we have $v[i] = u[i+\ell]$ for some
    $\ell \in \{1,2,\ldots,\abs{u}-1\}$. Therefore $u[2i] = u[i] = v[i - \ell] = v[2i - 2\ell] = u[2i - \ell]$ so, as
    the length of $u$ is odd, we conclude that $u[j] = u[j - \ell]$ for all $j \in \{0, 1, \ldots, \abs{u} - 1\}$. This
    implies that $u$ is a power of a word of length $\gcd(\ell, \abs{u})$, which contradicts the primitivity of $u$.

    Assume then that $v$ is not a product of the words $S$ and $L$. Consider an arbitrary infinite word $\infw{w}$ that
    is a product of the words $S$ and $L$ and has $u^5$ as a prefix. Suppose that $v$ occurs at position $\ell$ of
    $u^2$, so that the word $T^\ell(\infw{w})$ has $v^4$ as a prefix. Since $v$ is not a product of the words $S$ and
    $L$, the word $T^\ell(\infw{w})$ is not in $\{S, L\}^\omega$. Moreover, as $v$ is a solution to \eqref{eq:square},
    the words $\sqrt{T^\ell(\infw{w})}$ and $\sqrt[2]{T^\ell(\infw{w})}$ have $v$ as a prefix. Now
    \autoref{lem:monotone_or_periodic} implies that $\sqrt[2]{T^\ell(\infw{w})}$ is periodic. By
    \autoref{thm:nr_periodic_general}, the minimal period of $\sqrt[2]{T^\ell(\infw{w})}$ is conjugate to $S$. As
    $\abs{S}$ divides $\abs{v}$ and $v$ is primitive, we conclude that $\abs{v} = \abs{S}$. This contradicts the
    assumption that $u \notin \{S, L\}$.
  \end{proof}

  Notice that the final paragraph of the proof of \autoref{thm:conjugate_solutions} does not use the fact that $u$ is a
  solution to \eqref{eq:square}. Thus we obtain the following corollary.
  
  \begin{corollary}\label{cor:conjugate_solutions}
    If $u$ is not a product of the words $S$ and $L$ but is conjugate to such a product, then $u$ is not a solution to
    \eqref{eq:square}.
  \end{corollary}


  \begin{theorem}\label{thm:solution_characterization}
    If $u$ is a primitive solution to \eqref{eq:square} in $\Lang{\Omega}$ such that $\abs{u} \geq 2\abs{S}$, then
    $u = \g_k$ for some $k \geq 1$.
  \end{theorem}
  \begin{proof}
    Suppose that $u$ is a primitive solution to \eqref{eq:square} in $\Lang{\Omega}$ such that $\abs{u} \geq 2\abs{S}$.
    In particular, $u^2 \in \Lang{\Omega}$. First we aim to show that $\abs{u}$ is a multiple of $\abs{S}$. The word
    $u^2$ is a factor of a product of the words $S$ and $L$. If two consecutive words $S$ or $L$ from the product are
    completely contained in $u$, then $\abs{u}$ is a multiple of $\abs{S}$ by \autoref{lem:synchronization}. Suppose
    this is not the case. Since $\abs{u} \geq 2\abs{S}$, the word $u^2$ is a factor of a word
    $\diamond_1 \! \diamond_2 \! \diamond_3 \! \diamond_4 \! \diamond_5$, $\diamond_i \in \{S, L\}$, such that
    $\diamond_2$ and $\diamond_4$ are completely contained in $u$. It follows that $\diamond_4$ is an interior factor
    of $\diamond_1 \diamond_2$ or $\diamond_2 \diamond_3$. Consider the former case. Now also $\diamond_2$ is an
    interior factor of $\diamond_4 \diamond_5$. Observe that $\diamond_4$ cannot occur at position $1$ of
    $\diamond_1 \diamond_2$ because $\diamond_1$ and $\diamond_4$ may differ only by their first two letters; their
    common suffix of length $\abs{S} - 2$ would otherwise be unary. Thus $\diamond_2 \neq \diamond_4$ as otherwise
    $\diamond_2$ would be an interior factor of its square contradicting the primitivity of $\diamond_2$. Exactly
    symmetric argument shows that $\diamond_2 \neq \diamond_5$. By \autoref{lem:factorization_properties}, we see that
    $\diamond_4 = \diamond_5 = S$ and $\diamond_2 = L$. Further, we have $\diamond_1 = \diamond_3 = S$. The occurrence
    of $\diamond_2$ as an interior factor of $\diamond_4 \diamond_5$ must be followed by the first letter of
    $\diamond_2$. Hence the letter following the prefix $\diamond_1 \diamond_2$ of
    $\diamond_1 \! \diamond_2 \! \diamond_3 \! \diamond_4 \! \diamond_5$ must be the first letter of $\diamond_2$. This
    is a contradiction as $\diamond_2 \neq \diamond_3$. In the latter case, we analogously conclude that
    $\diamond_3 \neq \diamond_4$ and $\diamond_2 \neq \diamond_4$. This implies that
    $\diamond_2 = \diamond_3 = \diamond_5 = S$ and $\diamond_4 = L$. Again the occurrence of $\diamond_4$ as an
    interior factor of $\diamond_2 \diamond_3$ must be followed by the first letter of $\diamond_4$, but this is
    impossible as the prefix $\diamond_1 \! \diamond_2 \! \diamond_3 \diamond_4$ of
    $\diamond_1 \! \diamond_2 \! \diamond_3 \! \diamond_4 \! \diamond_5$ is followed by the first letter of
    $\diamond_5$. Thus we have shown that $\abs{u}$ is a multiple of $\abs{S}$.

    Now $u$ is conjugate to a product of the words $S$ and $L$. If $u$ is not itself a product of the words $S$ and
    $L$, then it is not a solution to \eqref{eq:square} by \autoref{cor:conjugate_solutions}. Therefore $u$ is a
    product of the words $S$ and $L$, and there exists a word $v$ in $\Lang{\Omega^*}$ such that $\sigma(v) = u$ and
    $v^2 \in \Lang{\Omega^*}$. \autoref{lem:square_conjugate} implies that $u$ is conjugate to $\g_k$ for some
    $k \geq 1$, and \autoref{thm:conjugate_solutions} shows that $u = \g_k$.
  \end{proof}

  It is certainly possible that $\Lang{\Omega}$ contains short solutions to \eqref{eq:square} that are not conjugates
  of $S$ or $L$. First of all, the solution $S$, as a reversed standard word, can have squares of shorter reversed
  standard words as factors; these are also solutions to \eqref{eq:square} in $\Lang{\Omega}$. Secondly, it is possible
  that there is a solution $u$ such that $\abs{S} < \abs{u} < 2\abs{S}$. For instance, if $S = 01010010$, then
  $\Lang{\Omega}$ contains the solution $01010010010$ of length $11$ as a factor of $SLSS$.

  \section{Further Remarks}\label{sec:further_remarks}
  In this paper, we constructed the fixed points $\infw{\Gamma}_1$ and $\infw{\Gamma}_2$ using the simple substitution
  $S \mapsto LS^{2\oc}$, $L \mapsto S^{2\oc+1}$. We observe that any word $\infw{w} = a_0 a_1 a_2 \cdots$ over
  $\{S, L\}$ satisfying $a_i = a_{2i}$ for all $i$ has the property $\sqrt{\infw{w}} = \infw{w}$ (see the remark of the
  final paragraph of \autoref{ssec:subshift_omega}). There are many other substitutions with suitable properties.
  Indeed, let $n$ be an odd positive integer, and consider a word $u$ of length $n$ over $\{S, L\}$ having the property
  that $\sqrt{u^2} = u$. For each multiplicative set of $\Z_n$ given by the element $2$, we can make an independent
  choice of a letter in $\{S, L\}$. For instance, if $n = 7$, then we have sets $\{0\}$, $\{1,2,4\}$, and $\{3,5,6\}$
  meaning that the word
  $\diamond_0 \! \diamond_1 \! \diamond_1 \! \diamond_2 \! \diamond_1 \! \diamond_2 \! \diamond_2$ is suitable for any
  choice of $\diamond_0$, $\diamond_1$, and $\diamond_2$ in $\{S, L\}$. Letting $\diamond_1 = S$ and $\diamond_2 = L$
  thus gives the substitution
  \begin{equation*}
    S \mapsto LSSLSLL, \quad L \mapsto SSSLSLL,
  \end{equation*}
  which again generates two fixed points of the square root map (after $S$ and $L$ are substituted by suitable
  solutions to \eqref{eq:square}). Mixing the applications of substitutions generated like this gives rise to even
  larger class of fixed points.

  It is quite unclear which of the results given generalize to this larger class of fixed points and the associated
  subshifts. \autoref{thm:periodic_finite_time} still applies, but it is unknown if variants of Theorems
  \ref{thm:injectivity} and \ref{thm:periodic_points} hold in general. It seems that a proof analogous to that of
  \autoref{thm:periodic_points} works in the case of a single generating substitution, but it is unclear if it works
  with two substitutions.

  Unlike in \autoref{thm:solution_characterization}, there can in general be other solutions to \eqref{eq:square} than
  those given directly by the generating substitution. This is because the substitution might itself have other
  solution generating patterns embedded. Consider for instance the substitution
  \begin{equation*}
    \tau\colon
    \begin{array}{l}
      S \mapsto LSSLSSLSS \\
      L \mapsto SSSLSSLSS
    \end{array}
  \end{equation*}
  and the associated subshift $\Omega$ with the periodic part adjoined. The images of the letters have the word
  $(LSS)^2$ as a suffix. Now $LSS$ is a solution to \eqref{eq:square}, so $\tau^k(LSS)$ is a solution to
  \eqref{eq:square} for all $k \geq 0$. Since $(\tau^k(LSS))^2 \in \Lang{\Omega}$ for all $k \geq 0$, we see that
  $\Lang{\Omega}$ contains arbitrarily long solutions to \eqref{eq:square} that are not equal to $\tau^k(S)$ or
  $\tau^k(L)$. Notice that these additional solutions cannot be used to produce additional fixed points as limits.

  There are many additional dynamical systems concepts that could be studied in our symbolic square root map setting.
  This time we finish our inquiry by making a remark on topological transitivity and topological mixing. A dynamical
  system $(X, f)$ is \emph{topologically transitive} if for every nonempty open sets $A$ and $B$ there exists an
  integer $n$ such that $f^n(A) \cap B \neq \emptyset$. The dynamical system $(\Omega, \sqrt{\cdot})$ is not
  topologically transitive because the words beginning with $0$ never map to words beginning with $1$ (the cylinder
  sets $[0]$ and $[1]$ are open sets); \autoref{prp:invariant_subsets} provides additional open sets with the same
  property. The same examples show that $(\Omega, \sqrt{\cdot})$ is not topologically mixing. A dynamical system
  $(X, f)$ is \emph{topologically mixing} if for any two nonempty open sets $A$ and $B$ there exists an integer $n$
  such that $f^N(A) \cap B \neq \emptyset$ for all $N \geq n$.

  \section{Open Problems}\label{sec:open_problems}
  In \autoref{sec:limit_set}, we gave some example values for the time $n$ it takes any word in
  $\overline{\{S, L\}^\omega} \setminus \{S, L\}^\omega$ to map to $S^\omega$ or $L^\omega$ in the case that $S$ is a
  reversed Fibonacci word. We conjecture that the quantity
  \begin{equation*}
    \log_2\left( \frac{1-\overline{\alpha}}{\min\{\abs{[S]},\abs{[L]}\}} \right)
  \end{equation*}
  given in \autoref{rem:bound_iterations_to_periodic} is close to the real value in general. Let us compute this
  quantity for the reversed Fibonacci words for comparison. In the case of the Fibonacci words, the slope $\alpha$ has
  continued fraction expansion $[0;2,1,1,1,\ldots]$, that is, $\alpha = 2 - \varphi$ where $\varphi$ is the golden
  ratio $(1+\sqrt{5})/2$. We replace the truncation $\overline{\alpha}$ by $\alpha$. It is not difficult to see that
  $\min\{\abs{[S]}, \abs{[L]}\} = \| F_k \alpha \|$ for a Fibonacci number $F_k$; here the norm $\|\cdot\|$ measures
  the distance to closest integer. From elementary properties of continued fractions, it is straigthforward to derive
  that $\| F_k \alpha \|^{-1} = \varphi F_k + F_{k-1}$ (save for some small values of $k$). Thus the above quantity now
  (approximately) equals
  \begin{equation*}
    \log_2 ( (\varphi - 1)(\varphi F_k + F_{k-1}) )
  \end{equation*}
  for $k$ such that $\abs{S} = F_k$. \autoref{tbl:fib_periodic_estimate} contains the values of this quantity for
  several $F_k$. Comparing these to the real values in \autoref{tbl:fib_periodic}, we see that at least in this special
  case our conjecture seems valid. The values of \autoref{tbl:fib_periodic} seem to match the sequence
  \href{https://oeis.org/A020909}{A020909} in Sloane's \emph{On-Line Encyclopedia of Integer Sequences} \cite{oeis}.
  This sequence gives the number of bits in the base 2 representations of the Fibonacci numbers suggesting that the
  real value is $1 + \log_2 F_{k-1}$. The difficulty here lies in the first part of the proof of
  \autoref{thm:periodic_finite_time} where we showed that eventually $\sqrt[n]{\infw{w}}$ must be periodic. Unless we
  know that $\sqrt[n]{\infw{w}}$ is periodic, we cannot transfer the situation to the system of rational rotations
  where it is easy to estimate how long it takes for a word to map to $S^\omega$ or $L^\omega$. Based on our computer
  experiments, it seems that typically $\infw{w}$ maps to a periodic quite early suggesting that the above quantity
  should be close to the truth. We have observed examples where it takes a longer time for $\infw{w}$ to become
  periodic, but in these cases the interval corresponding to the minimal period is already quite close to $[S]$ or
  $[L]$ on the circle balancing the situation. Overall, we do not have a clear picture of the situation. We propose the
  following open problem.

  \begin{table}
  \centering  
  {\footnotesize
  \begin{tabular}{c *{16}{@{\hspace{3.1mm}}c}}
    $\abs{S}$ \hspace{2mm} & 8    & 13    & 21    & 34    & 55    & 89    & 144  & 233  & 377  & 610  & 987   & 1597  & 2584  & 4181  & 6765 \\
    \hline \\
    $n$ \hspace{2mm}       & 3.47 & 4.16  & 4.85  & 5.55  & 6.24  & 6.94  & 7.63 & 8.33 & 9.02 & 9.71 & 10.41 & 11.11 & 11.80 & 12.50 & 13.19 \\
    \hline
  \end{tabular}}
  \caption{Estimation of the relation of $\abs{S}$ and $n$ of \autoref{thm:periodic_finite_time} when $S$ is a reversed
  Fibonacci word. The first two digits are correct.}\label{tbl:fib_periodic_estimate}
  \end{table}

  \begin{problem}
    Prove a good estimate on the number of steps required for any word in
    $\overline{\{S, L\}^\omega} \setminus \{S, L\}^\omega$ to map to $S^\omega$ or $L^\omega$.
  \end{problem}

  Regarding the preimages of the word in the periodic part $\Omega_P$ we left uncharacterized, it would be interesting
  to know which words of $\Omega_P$ are images of words in $\Omega_A$. Thus we propose the following open problem.

  \begin{problem}
    Characterize the set $\sqrt{\Omega} \setminus \Omega_A$. How large is this set?
  \end{problem}

  In the case that $S$ is a reversed Fibonacci word, it seems that
  \begin{equation*}
    \abs{\sqrt{\Omega} \setminus \Omega_A} \in \{\abs{S}/2, (\abs{S}+1)/2, (\abs{S}-1)/2\},
  \end{equation*}
  so the size of $\sqrt{\Omega} \setminus \Omega_A$ seems to be approximately half of the size of $\Omega_P$.

  \section*{Acknowledgments}
  The work of the first author was supported by the Finnish Cultural Foundation by a personal grant. He also thanks the
  Department of Computer Science at Åbo Akademi for its hospitality. The second author was partially supported by the
  Vilho, Yrjö and Kalle Väisälä Foundation. Jyrki Lahtonen deserves our thanks for fruitful discussions.

  \printbibliography
  
\end{document}